\documentclass[10pt,draft,reqno]{amsart}
     \makeatletter
     \def\section{\@startsection{section}{1}%
     \z@{.7\linespacing\@plus\linespacing}{.5\linespacing}%
     {\bfseries
     \centering
     }}
     \def\@secnumfont{\bfseries}
     \makeatother
   \usepackage{amssymb,latexsym}
\newtheorem{theorem}{Theorem}[section]
\DeclareMathOperator\supp{supp}
\newtheorem{lemma}[theorem]{Lemma}

\theoremstyle{definition}

\theoremstyle{remark}
\newtheorem{remark}[theorem]{Remark}
\numberwithin{equation}{section}
\makeatletter
\newcommand\@received{Received 2024-5-6;
  Accepted 2024-6-2; Communicated by the editors.}
\renewcommand\@adminfootnotes
{\let \@makefnmark \relax \let \@thefnmark \relax
\@footnotetext{\@received }
\ifx \@empty \@date \else \@footnotetext {\@setdate }\fi
\ifx \@empty \@subjclass \else \@footnotetext {\@setsubjclass }\fi
\ifx \@empty \@keywords \else \@footnotetext {\@setkeywords }\fi
\ifx \@empty \thankses \else \@footnotetext
{\def \par {\let \par \@par }\@setthanks }\fi}
\makeatother

\begin{document}

\title[Holomorphic Functional Calculus approach]{Holomorphic Functional Calculus approach\\
to the  Characteristic Function of\\ Quantum Observables}
\author[Andreas Boukas]{Andreas Boukas*}
\thanks{* Corresponding author}
\address{Andreas Boukas: Centro Vito Volterra, Universit\`{a} di Roma Tor Vergata, via Columbia  2, 00133 Roma,
Italy} \email{andreasboukas@yahoo.com}



\begin{abstract}
We show how Cauchy's  Integral Formula and the ideas of Dunford's
Holomorphic Functional Calculus (for unbounded operators) can be
used  to compute  the Vacuum Characteristic Function (Quantum
Fourier Transform) of quantum random variables defined as
self-adjoint operators on $L^2(\mathbb{R},\mathbb{C})$. We
consider in detail several quantum observables defined in terms of
the position and momentum operators   $X$, $P$, respectively, on
$L^2(\mathbb{R},\mathbb{C})$.
\end{abstract}

\maketitle

\section{Introduction}\label{intro}

Identifying the statistical properties of a \textit{quantum random
variable}, typically by computing its \textit{vacuum
characteristic function}, is of primary interest in
\textit{quantum probability} and \textit{quantum stochastic
analysis}. This is usually done through Lie algebraic means by
expressing the observable as a linear combination of the
generators of a naturally associated Lie algebra. This approach
has been studied extensively, for example in \cite{AccBouCOSA}. In
\cite{BFspec2} and \cite{BFspec} we showed how it can be done
through analytic methods, with the use of \textit{von Neumann's
spectral theorem} for unbounded self-adjoint operators on  Hilbert
spaces and also with the use of \textit{Stone's formula} for the
spectral resolution of an unbounded self-adjoint operator on a
Hilbert space. In this paper, which can be viewed as Part III of
\cite{BFspec2} and \cite{BFspec}, we show how to compute the
vacuum characteristic function of a quantum observable with the
use of \textit{holomorphic functional calculus} and
\textit{Cauchy's integral formula} on a suitable contour.

 As in \cite{BFspec2} and  \cite{BFspec}, we consider the self-adjoint (see \cite{QMI}, Section 2.3) \textit{position, momentum} and
 \textit{identity} operators, $X, P$ and $\mathbf{1}$, respectively, defined in $L^2(\mathbb{R},\mathbb{C})$ with inner
product $\langle \cdot, \cdot \rangle$ and norm $\|\cdot\|$,
respectively,
\begin{equation}
\langle f, g \rangle
=\frac{1}{\sqrt{\hslash}}\,\int_{\mathbb{R}}\,\overline{f(x)}g(x)\,dx
\ \,,\,\,\|f\|=\left(
\frac{1}{\sqrt{\hslash}}\,\int_{\mathbb{R}}\,|f(x)|^2\,dx
\right)^{1/2} \  ,
\end{equation}
 by
\begin{equation}
X\,f(x)=x\,f(x)\,\,;\,\,
P\,f(x)=-i\,\hslash\,f^{\prime}(x)\,\,;\,\, \mathbf{1}\,f(x)=f(x)\
,
\end{equation}
and satisfying the commutation relations
\begin{equation}
\lbrack P, X \rbrack=-i\hslash\mathbf{1}\  ,
\end{equation}
on
\begin{equation}
\Omega={\rm dom }(X)\cap {\rm dom }(P)\  ,
\end{equation}
 where,
\begin{equation}
{\rm dom }(X)=\{f\in L^2(\mathbb{R},\mathbb{C}))
\,:\,\int_{\mathbb{R}}x^2\,|f(x)|^2\,dx<\infty\}\  ,
\end{equation}
and
\begin{equation}
{\rm dom }(P)=\{f\in L^2(\mathbb{R},\mathbb{C}) \,:\,\mbox{$f$ is
abs. cont. and }\int_{\mathbb{R}}\,\left|f^\prime(x)\right|^2\,dx<
\infty\}\  ,
\end{equation}
are respectively the, dense in $L^2(\mathbb{R},\mathbb{C})$,
domains of $X$ and $P$ .  Here,``abs. cont" stands for
\textit{absolutely continuous} and the existence of $f^{\prime}$
is only required \textit{almost everywhere}.  To show that $f$ is
absolutely continuous it suffices to show that $f^{\prime}$
\textit{is in $L^2(\mathbb{R},\mathbb{C})$ as a distribution} (see
\cite{Richtmyer}, p.88 and p.132), meaning that there exists $M>0$
such that, for all $\phi \in C_0^{\infty}(\mathbb{R})$,
\begin{equation}\label{acd}
|\langle f^{\prime}, \phi \rangle|
=\left|\int_{\mathbb{R}}f^{\prime}(s) \phi(s)\,ds\right|
 \leq
 M\, \|\phi \|   \ .
\end{equation}
 Functions in the domain of $P$
are continuous and vanish at infinity. Since it contains
$C_0^{\infty}(\mathbb{R})$ (the set of \textit{bump} functions),
$\Omega$ is nonempty and dense in $L^2(\mathbb{R},\mathbb{C})$.
Moreover, functions in $\Omega$ vanish at infinity. The subset of
$C_{0}^{\infty}(\mathbb{R})$ consisting of bump functions
vanishing in a neighborhood of zero is also dense in
$L^2(\mathbb{R},\mathbb{C})$ (see Lemma \ref{den} for a proof).

The \textit{Schwartz class} $\mathcal{S}$ is a common invariant
domain of $X$ and $P$ which is also dense in
$L^2(\mathbb{R},\mathbb{C})$ and contains
$C_0^{\infty}(\mathbb{R})$. Whenever sums or products of $X$ and
$P$ are considered, $\mathcal{S}$ will be taken as their domain.
Since $\mathcal{S}$ is  viewed as a subset of
$L^2(\mathbb{R},\mathbb{C})$, we can allow in $\mathcal{S}$
functions with a finite number of finite-jump discontinuities.
This is used, in particular, in Theorem \ref{csp}. The Schwartz
class and its properties, as related to Quantum Mechanics, are
presented in \cite{BeSen}. As in \cite{BFspec2} and \cite{BFspec},
we normalize to $\hslash=1$.

The function
\begin{equation}\label{phi}
\Phi=\Phi(x)=
\pi^{-1/4}\,e^{-\frac{x^2}{2\hslash}}=\pi^{-1/4}\,e^{-\frac{x^2}{2}}\
,
\end{equation}
is a unit vector in $\mathcal{S}$ . For $z\in\mathbb{C}$ we denote
\begin{equation}
R(z;T)=(z-T)^{-1} \ ,
\end{equation}
 the \textit{resolvent} of $T$. The
\textit{vacuum characteristic function} of $T$ is
\begin{equation}
 \langle  e^{itT} \rangle=  \langle \Phi, e^{itT}\Phi \rangle\ .
\end{equation}

For $a, b\in\mathbb{R}$ we denote by $\delta_{a,b}$,  $\delta_a$
and $H_a$ , respectively, the \textit{Kronecker delta},
\textit{Dirac delta} and \textit{Heaviside unit step} functions
defined, for a \textit{test function $f$}, by
\begin{equation}
\delta_{a,b}=\left\{
\begin{array}{llr}
 1\   ,&   a=b  \\
& \\
0\   ,& a\neq b
 \end{array}
 \right.  \  ,
\end{equation}
\begin{equation}
\int_{\mathbb{R}} f(x)\delta_a (x)\,dx =\int_{\mathbb{R}}
f(x)\delta (x-a)\,dx=f(a) \ ,
\end{equation}
 and
\begin{equation}
H_a(x)=H(x-a)=\left\{
\begin{array}{llr}
 1\   ,&  x \geq a  \\
& \\
0\   ,& x < a
 \end{array}
 \right.  \  .
\end{equation}
For a complex number $z$ we denote ${\rm Re }z$, ${\rm Im }z$, its
real and imaginary part, respectively, and we denote the
\textit{support} of a function $g$ by $\supp(g)$.

In  Lemma 12A, p.43 of \cite{goldberg}, it is shown that for any
real numbers $\epsilon >0$ and $\alpha$ we have
\begin{equation}\label{i1}
\int_{-\infty}^\infty e^{i \alpha t} e^{- \epsilon t^2}\,
dt=\left( \frac{\pi}{\epsilon}
\right)^{\frac{1}{2}}e^{-\frac{\alpha^2}{4\epsilon}} \ .
\end{equation}
Formula (\ref{i1}) is also valid for $\epsilon\in\mathbb{C}$ with
positive real part and it is used throughout this paper.

As in \cite{BFspec}, we define the Fourier transform of $f$ by
\begin{equation}\label{FT}
(Uf)(t)=\hat{f}(t)=(2 \pi)^{-1/2}\,\int_{-\infty}^\infty
\,e^{i\lambda t}f(\lambda)d\lambda\  ,
\end{equation}
and the inverse Fourier transform of $\hat{f}$ by
\begin{equation}\label{IFT}
(U^{-1}\hat{f})(\lambda)=f(\lambda)=(2 \pi)^{-1/2}\,\int_{-\infty}^\infty \,e^{-i\lambda
t}\hat{f}(t)dt\ .
\end{equation}
 The definition of the Fourier transform in (\ref{FT}) is the one used in p.315 of \cite{yosida} (and also in \cite{goldberg}) and it is compatible with the definition of the
characteristic function of a random variable with probability
density function $f$ (\cite{lucacs}, p.10,  formula (1.3.6)).

The operators $X$ and $P$ are \textit{unitarily equivalent}
through the Fourier transform (see \cite{yosida}), i.e.,
\begin{equation}\label{UE}
Xf=U^{-1}PUf \,\,,\,\,Pf=UXU^{-1}f  \,,\, f\in\mathcal{S} \ .
\end{equation}
Equation (\ref{UE}) is generally true if $f$ is such that
\begin{equation}\label{UEC}
f \in \Omega \,,\,Uf \in {\rm dom }(P)\,,\, U^{-1}f \in {\rm dom
}(X) \ .
\end{equation}

 Interpreting \textit{Cauchy's integral formula} for $e^{itT}$,
as in \textit{Dunford's Holomorphic Functional Calculus} for
unbounded operators (see \cite{DS1} p.599) in the weak sense, we
may define the vacuum characteristic function of a self-adjoint
operator $T$ by
\begin{equation}\label{tw}
\langle  e^{itT} \rangle=\frac{1}{2\pi i}\oint_C e^{itz} \langle
\Phi, R(z;T)\Phi \rangle \,dz\  ,
\end{equation}
where the closed contour $C$ is, in principle, large enough to
contain the singularities of
\begin{equation}
F(z)=e^{itz} \langle \Phi, R(z;T)\Phi \rangle= e^{itz}
\int_{-\infty}^{\infty}
 \Phi(s) R(z;T)\Phi(s) \,ds  \ .
\end{equation}
 In this paper we will consider
\begin{equation}
T\in\left\{X, P, X+P, XP+PX,
\frac{1}{2}\left(X^2+P^2\right)\right\}\  ,
\end{equation}
so $T$, being (at least) symmetric and densely defined, will be
\textit{closable} \cite{Konrad}. However, the function $f(z)=
e^{itz} $ is not \textit{holomorphic at infinity} (meaning,
$f(1/z)=e^{it/z}$ is not holomorphic at $0$), so Dunford's formula
\begin{equation}
 f(T):=\frac{1}{2\pi i}\oint_C f(z) R(z;T)\,dz\  ,
\end{equation}
cannot be  directly applied to define $e^{itT}$, and we will use
its weak version, in the vacuum state $\Phi$,  instead. However,
even in simple cases, the spectrum of $T$ covers the entire real
line. We will show how the contour integral in (\ref{tw}) should
be interpreted.

Passage of the limit under the integral sign is typically
justified by appealing to \textit{ Lebesgue's bounded convergence
theorem} as, for example in Theorem 9.1 of \cite{BFspec}.

Though, the main theme of this paper is the computation of the
vacuum characteristic function of $X, P, X+P, XP+PX$ and
$\frac{1}{2}(X^2+P^2)$, their relevant spectral properties are
studied in detail. These operators are well-studied in quantum
mechanics. When something is mentioned in the literature, a
reference is given. The complete proofs are, however, included for
completeness and they were given by the author.

\section{The Vacuum Characteristic Function of $X$  }

\begin{theorem}\label{sX} The operator  $X$ has only continuous spectrum
consisting of the entire
real line. Moreover, for $z\in \mathbb{C}\setminus \mathbb{R}$
and  $g$ in the range of $z-X$, the resolvent operator $R(z; X)$
is defined by
\begin{equation}\label{ffr}
R(z; X)g(s)=\frac{g(s)}{z-s} \  .
\end{equation}
\end{theorem}
\begin{proof}Since $X$ is self-adjoint, its spectrum is real. If $z\in\mathbb{R} $
is an eigenvalue of $X$ then a corresponding
eigenfunction $G$ would satisfy
\begin{equation}
(s-z)G(s)=0 \,,\, s\in\mathbb{R}\  ,
\end{equation}
which would imply that $G=0$ almost everywhere. Thus $X$ has no
point spectrum.

\medskip

To see that, for  $z\in\mathbb{R}$, the resolvent operator
$R(z;X)$ is not bounded, we notice that, for $n\in\{1, 2,...\}$,
the function (see \cite{gdansk})
\begin{equation}
G_n(s)=\frac{1}{z-s}\chi_{{}_{\left\lbrack z-1, z-\frac{1}{n}
\right\rbrack}}(s)=\left\{
\begin{array}{llr}
\frac{1}{z-s}  , &\mbox{ if  }  z-1\leq s \leq z-\frac{1}{n}  \\
& \\
 0 ,&\mbox{ otherwise  }
\end{array}
\right. \ ,
\end{equation}
is in the domain of $X$ and the function
\begin{equation}
g_n(s)=(z-s)G_n(s)=\chi_{{}_{\left\lbrack z-1, z-\frac{1}{n}
\right\rbrack}}(s)=\left\{
\begin{array}{llr}
1  , &\mbox{ if  }  z-1\leq s \leq z-\frac{1}{n}  \\
& \\
 0 ,&\mbox{ otherwise  }
\end{array}
\right.  \ ,
\end{equation}
is, by construction, in the range of $z-X$ (therefore in the
domain of $R(z, X)$) with corresponding $L^2$-norms
\begin{equation}
\|R(z;X)g_n\|^2=\|G_n\|^2=n-1 \to \infty , \,\,  n \to \infty \ ,
\end{equation}
and
\begin{equation}
\|g_n\|^2= 1-\frac{1}{n}\to 1, \,\, n \to \infty \  .
\end{equation}
Thus, there can be no non-negative real number $M$ such that
\begin{equation}
\|R(z;X)g_n\|\leq M \|g_n\|, \,\, n\in\{1, 2,...\} \ .
\end{equation}
Since $X$ is self-adjoint, its residual spectrum is empty
(\cite{Richtmyer}, Theorem 1). Therefore, every $z\in\mathbb{R}$
is in the continuous spectrum of $X$.

Alternatively, to show that the range of $z-X$ is dense in $L^2(\mathbb{R},\mathbb{C})$,
we notice that for $f\in  L^2(\mathbb{R},\mathbb{C})$, the sequence $(f_n)$ defined by
(see \cite{gdansk})
\begin{equation}\label{deffn}
f_n(s) = \frac{f(s)}{z-s}\, \chi_{{}_{ \left(z-\frac{1}{n},
z+\frac{1}{n}\right)^c}}(s)=\left\{
\begin{array}{llr}
\frac{f(s)}{z-s}  , &\mbox{ if  }   s \notin \left(z-\frac{1}{n}, z+\frac{1}{n}\right)  \\
& \\
 0 ,&\mbox{ otherwise  }
\end{array}
\right.\  ,
\end{equation}
is in the domain of $X$ and $(z-X)f_n\to f$, in the $L^2$-sense,
as $n\to\infty$.

\medskip

Finally, for (\ref{ffr}), for $z\in \mathbb{C}\setminus
\mathbb{R}$ and $s\in\mathbb{R}$, we have
\begin{equation}
R(z; X)g(s)=G(s) \iff (z-X)G(s)=g(s) \iff G(s)=\frac{g(s)}{z-s} \
.
\end{equation}
\end{proof}

\begin{theorem}\label{X} For $t\in \mathbb{R}$, the vacuum characteristic function of
 $X$  is
\begin{equation}
\langle  e^{i t X}  \rangle= e^{-\frac{t^2}{4}} \ .
\end{equation}
\end{theorem}
\begin{proof} By (\ref{tw}) we have,
\begin{align}
\langle  e^{itX} \rangle&=\frac{1}{2\pi i}\oint_C e^{itz} \langle
\Phi, R(z;X)\Phi
\rangle \,dz \\
&=\frac{1}{2\pi i}\oint_C e^{itz}\int_{\mathbb{R}} \Phi(s) R(z,
X)\Phi(s)\,ds \,dz \notag\\
&=\frac{1}{2\pi^{3/2} i}\oint_C e^{itz}\int_{\mathbb{R}}
\frac{e^{-s^2} }{z-s}\,ds \,\,dz\notag\\
&=\frac{1}{2\pi^{3/2} i}\oint_C\int_{\mathbb{R}} \frac{e^{itz-s^2}
}{z-s}\,ds \,\,dz\notag\\
&=\frac{1}{2\pi^{3/2} i}\lim_{\varepsilon\to 0^+}\lim_{r\to
\infty}\oint_{C_{r+\varepsilon}} \int_{-r}^r \frac{e^{itz-s^2}
}{z-s}\,ds \,\,dz \notag \  ,
\end{align}
where,
\begin{equation}
C_{r+\varepsilon}=\{z\in\mathbb{C} \,:\, |z|=r+\varepsilon  \} \ .
\end{equation}
Thus, by Cauchy's integral formula,
\begin{align}
\langle  e^{itX} \rangle&=\frac{1}{2\pi^{3/2}
i}\lim_{\varepsilon\to 0^+}\lim_{r\to \infty} \int_{-r}^r
\oint_{C_{r+\varepsilon}}\frac{e^{itz-s^2}
}{z-s}\,dz \,\,ds\\
&=\frac{1}{2\pi^{3/2} i}\lim_{\varepsilon\to 0^+}\lim_{r\to
\infty}\int_{-r}^re^{-s^2}\oint_{C_ {r+\varepsilon }}
\frac{e^{itz} }{z-s}\,dz \,\,ds\notag\\
&=\frac{1}{2\pi^{3/2} i}\lim_{r\to \infty}\int_{-r}^re^{-s^2}\, 2\pi i\,e^{its} \,ds\notag\\
&=\frac{1}{\pi^{1/2} }\int_{\mathbb{R}}e^{its-s^2} \,ds
=\frac{1}{\pi^{1/2} } \pi^\frac{1}{2} e^{-\frac{t^2}{4} } =
e^{-\frac{t^2}{4} } \notag\  .
\end{align}
\end{proof}

\section{The Vacuum Characteristic Function of $P$  }

\begin{theorem}\label{Pna}  The operator  $P$ has only continuous spectrum consisting
of the entire
real line. Moreover, for $g$ in the range of $z-P$, the resolvent
operator $R(z; P)$ is
\begin{equation}
R(z;P)g(s)=\left\{
\begin{array}{llr}
-i\int_{-\infty}^s e^{iz(s-w)}g(w) \,dw , &\,\, {\rm Im}\,z> 0   \\
& \\
 i\int^{\infty}_s e^{iz(s-w)}g(w) \,dw ,&\,\, {\rm
Im}\, z < 0
\end{array}
\right. \ .
\end{equation}
For $g$ not identically equal to zero almost everywhere,
$R(z;P)g(s)$ is not a continuous (thus not analytic) function of
$z\in\mathbb{C}$, at real $z$'s, for almost all $s$.
\end{theorem}
\begin{proof} Since $P$ is self-adjoint, its spectrum is real. If $z\in\mathbb{R}$ is
an eigenvalue of $P$, then a corresponding
eigenfunction $G$ would satisfy
\begin{equation}
i\,G^{\prime}(s)+zG(s)=0 \,,\, s\in\mathbb{R}\ ,
\end{equation}
which would imply that
\begin{equation}
G(s)=c e^{izs}\, , c\in\mathbb{C}\ ,
\end{equation}
which is in $L^2(\mathbb{R},\mathbb{C})$ if and only if $c=0$,
i.e., if and only if $G=0$. Thus $P$ has no point spectrum.

\medskip

To see that every real number $z$ is in the continuous spectrum of
$P$ we consider (see \cite{Richtmyer}, p.190) the function
\begin{equation}
g_\beta(s)= \beta^{1/4} e^{-\beta s^2}e^{i z s}\,,\,\beta>0 \ ,
\end{equation}
and let
\begin{equation}
f_\beta(s)=(z-P)g_\beta(s)=-2 i \beta^{5/4} s e^{-\beta s^2}e^{i z
s} \  .
\end{equation}
Then,
\begin{equation}
 \|f_\beta\|^2=   \|(z-P)g_\beta\|^2=\beta \,\sqrt{\frac{\pi}{2}}   \to 0 \mbox{ as } \beta\to
0 \ ,
\end{equation}
and
\begin{equation}
\|R(z; P)f_\beta\|^2=\|g_\beta\|^2=\sqrt{\frac{\pi}{2}} \  .
\end{equation}
Thus, there can be no non-negative real number $M$ such that, for
all $ \beta>0$,
\begin{equation}
\|R(z; P)f_\beta\|\leq M \|f_\beta\| \ .
\end{equation}
Thus $R(z; P)$ is not continuous for any real $z$. Since $P$ is
self-adjoint, its residual spectrum is empty (\cite{Richtmyer},
Theorem 1). Therefore, every $z\in\mathbb{R}$ is in the continuous
spectrum of $P$.

Alternatively, to show that the range of $z-P$ is dense in $L^2(\mathbb{R},\mathbb{C})$,
 let $U$ and $U^{-1}$ denote the Fourier transform and its inverse as in (\ref{FT})
 and (\ref{IFT}), let $ g, f\in \mathcal{S}$ with $g=Uf$, and let
  $(f_n)$ be the sequence defined in (\ref{deffn}) with $(z-X)f_n\to f$, in the $L^2$-sense,
  as $n\to\infty$. Then, by (\ref{UE}),  $U^{-1}(z-P)Uf_n\to f$,
   in the $L^2$-sense, as $n\to\infty$ . Thus, the sequence
\begin{align}
g_n(s)&:=Uf_n(s)=\hat{f_n}(s)=(2
\pi)^{-1/2}\,\int_{-\infty}^\infty
\,e^{i\lambda s}f_n(\lambda)\,d\lambda \label{tb}\\
&=(2 \pi)^{-1/2}\,\int_{\mathbb{R}-\left(z-\frac{1}{n},
z+\frac{1}{n}\right)} \,\frac{e^{i\lambda s}}{z-\lambda}
 f(\lambda)\,d\lambda \notag\\
&=(2 \pi)^{-1/2}\,\int_{\mathbb{R}-\left(z-\frac{1}{n},
z+\frac{1}{n}\right)} \,\frac{e^{i\lambda s}}{z-\lambda}
  (U^{-1}g)(\lambda)\,d\lambda \notag\\
  &=(2 \pi)^{-1}\,\int_{\mathbb{R}-\left(z-\frac{1}{n},
z+\frac{1}{n}\right)} \int_{\mathbb{R}}\frac{e^{i\lambda
(s-w)}}{z-\lambda} g(w)\,dw \, d\lambda \ , \notag
\end{align}
is in the domain of $P$ and,  by the continuity of the Fourier
transform,  $U^{-1}(z-P)Uf_n\to f$ implies $(z-P)g_n\to g$, in the
$L^2$-sense, as $n\to\infty$.

\medskip

 For $z\in\mathbb{C}\setminus \mathbb{R}$ and  $s\in\mathbb{R}$, as shown in
 \cite{BFspec} (see also \cite{Richtmyer} p. 191),
\begin{equation}\label{res}
G(s):=R(z;P)g(s)=\left\{
\begin{array}{llr}
-i\int_{-\infty}^s e^{iz(s-w)}g(w) \,dw  , &\,\, {\rm Im}\,z> 0   \\
& \\
 i\int^{\infty}_s e^{iz(s-w)}g(w) \,dw ,&\,\,{\rm
Im}\, z<0
\end{array}
\right. \ .
\end{equation}
If ${\rm Im}\, z=0$ then  $G$ cannot be  in the domain of $P$,
since then
\begin{equation}
\lim_{t\to \pm\infty}|e^{-izt} G(t)|=\lim_{t\to \pm\infty}
|G(t)|=0 \  ,
\end{equation}
which, as  in the proof of Proposition 9.2 of \cite{BFspec}, would
allow us to integrate from $-\infty$ to $s$ and also from $s$ to
$\infty$ and get two expressions for $G(s)$, namely
\begin{equation}
G(s)=\left\{
\begin{array}{llr}
-i\int_{-\infty}^s e^{iz(s-w)}g(w) \,dw  \\
& \\
 i\int^{\infty}_s e^{iz(s-w)}g(w) \,dw
\end{array}
\right. \ ,
\end{equation}
which, after equating them, would imply that
\begin{equation}
\int_{-\infty}^\infty e^{iz(s-w)}g(w) \,dw=0 \ ,
\end{equation}
which is not true in general. For example, for $g=\Phi$,
\begin{equation}
\int_{-\infty}^\infty e^{iz(s-w)}g(w)
\,dw=\sqrt{2}\pi^{1/4}e^{isz-\frac{z^2}{2}}\neq 0 \  .
\end{equation}
Thus for $g$ not identically equal to zero almost everywhere,
$R(z;P)g(s)$ is not a continuous (thus not analytic) function of
$z\in\mathbb{C}$, at real $z$'s, for almost all $s$.
\end{proof}

\begin{remark}\label{rem} \rm
The appeal to (\ref{UE}) in the proof of Theorem \ref{Pna} can be
justified as follows. We need to show that, for $n\in\mathbb{N}$
and $z\in\mathbb{R}$, the function $f_n$ defined by
\begin{equation}
f_n(s)= \frac{f(s)}{z-s}\, \chi_{{}_{ \left(z-\frac{1}{n},
z+\frac{1}{n}\right)^c}}(s) \  ,
\end{equation}
has the properties: $f_n \in\Omega  \, , \, Uf_n\in{\rm dom }(P) \, , \,
U^{-1}f_n\in{\rm dom }(X) $.

To show that $f_n \in\Omega$  we have:
\begin{align}
\int_{-\infty}^\infty|f_n(s)|^2\,ds&=\int_{-\infty}^{z-\frac{1}{n}}
\frac{|f(s)|^2}{|z-s|^2} \,ds+\int_{z+\frac{1}{n}}^{\infty}
\frac{|f(s)|^2}{|z-s|^2} \,ds\\
&\leq n^2 \left( \int_{-\infty}^{z-\frac{1}{n}} |f(s)|^2
\,ds+\int_{z+\frac{1}{n}}^{\infty} |f(s)|^2 \,ds \right)\notag\\
 &\leq n^2 \|f\|^2 <\infty \  ,\notag
\end{align}
so $f_n\in L^2(\mathbb{R},\mathbb{C})$.

The derivative of $f_n$  exists almost everywhere, i.e. at all
$s\in\mathbb{R}\setminus \{z\pm \frac{1}{n}\}$, and is given by
\begin{equation}
f_n^\prime (s)=\left\{
\begin{array}{llr}
\frac{f^\prime (s)}{z-s}+\frac{f (s)}{(z-s)^2}  \,,\,  s \notin \left[z-\frac{1}{n}, z+\frac{1}{n}\right]  \\
& \\
 0 \,,\,  s \in \left(z-\frac{1}{n}, z+\frac{1}{n}\right)
\end{array}
\right. \  .
\end{equation}
For $\phi \in C_0^{\infty}(\mathbb{R})$ with $\supp (\phi)=K$ and
$\rho=\max \{|\max K|, |\min K| \}$,
\begin{align}
&|\langle f_n^{\prime}, \phi \rangle|
=\left|\int_{\mathbb{R}}f_n^{\prime}(s) \phi(s)\,ds\right|
 \leq \int_{\mathbb{R}}|f_n^{\prime}(s) \phi(s)|\,ds\\
 &\leq \int_{K\cap \left[z-\frac{1}{n}, z+\frac{1}{n}\right]^c }\frac{|f(s)-s f^{\prime}(s)|}{|z-s|^2 }\, |\phi(s)|\,ds
 \notag\\
 &\leq n^2 \int_{K\cap \left[z-\frac{1}{n}, z+\frac{1}{n}\right]^c }|f(s)-s f^{\prime}(s)|\,
 |\phi(s)|\,ds\notag \\
&\leq n^2\left( \int_{K\cap \left[z-\frac{1}{n},
z+\frac{1}{n}\right]^c }|f(s)-s
f^{\prime}(s)|^2\,ds\right)^{1/2}\,\left(\int_{K\cap
\left[z-\frac{1}{n}, z+\frac{1}{n}\right]^c
}|\phi(s)|^2\,ds\right)^{1/2}\notag\\
&\leq n^2\left( \int_{K\cap \left[z-\frac{1}{n},
z+\frac{1}{n}\right]^c }|f(s)-s
f^{\prime}(s)|^2\,ds\right)^{1/2}\,\|\phi\|\notag\\
&\leq n^2\left( \int_{K\cap \left[z-\frac{1}{n},
z+\frac{1}{n}\right]^c }\left(2|f(s)|^2+2 |s|^2
|f^{\prime}(s)|^2\right)\,ds\right)^{1/2}\,\|\phi\|\notag\\
&\leq n^2 \sqrt{2}  \left( \int_{K\cap \left[z-\frac{1}{n},
z+\frac{1}{n}\right]^c }\left(|f(s)|^2+ \rho^2
|f^{\prime}(s)|^2\right)\,ds\right)^{1/2}\,\|\phi\|\notag\\
&\leq  M\, \|\phi \| \  ,\notag
\end{align}
where
\begin{equation}
M=n^2 \sqrt{2}  \left(\|f\|+\rho^2 \|f^\prime \| \right)\ ,
\end{equation}
so $f_n^{\prime}$ is in $L^2(\mathbb{R},\mathbb{C})$ as a
distribution, so  $f_n$ is absolutely continuous.

Moreover, as in the proof of $f_n\in L^2(\mathbb{R},\mathbb{C})$
and the absolute continuity of $f_n$,

\begin{equation}\label{bb}
\int_{\mathbb{R}}s^2|f_n(s)|^2\,ds\leq n^2
\int_{\mathbb{R}}s^2|f(s)|^2\,ds<\infty \ ,
\end{equation}
since $f\in\mathcal{S}$ implies $f\in {\rm dom }(X)$, and
\begin{align}
\int_{\mathbb{R}}|f_n^\prime(s)|^2\,ds&=\int_{
\left[z-\frac{1}{n}, z+\frac{1}{n}\right]^c }\frac{|f(s)-s
f^{\prime}(s)|*2}{|z-s|^2 }\,ds \\
&\leq n^4 \int_{ \left[z-\frac{1}{n}, z+\frac{1}{n}\right]^c
}|f(s)-s
f^{\prime}(s)|^2\,ds  \notag\\
& \leq n^4 \int_{ \left[z-\frac{1}{n}, z+\frac{1}{n}\right]^c
}\left(|f(s)|+|s| \,
|f^{\prime}(s)|\right)^2\,ds  \notag\\
 &\leq 2 n^4 \int_{ \left[z-\frac{1}{n}, z+\frac{1}{n}\right]^c
}\left(|f(s)|^2+|s|^2 \,
|f^{\prime}(s)|^2\right)\,ds  \notag\\
&\leq 2 n^4 \left(\int_{\mathbb{R}}|f(s)|^2\,ds
+\int_{\mathbb{R}}s^2|f(s)|^2\,ds \right) <\infty \notag \ ,
\end{align}
since  $f\in\mathcal{S}$ implies $f\in \Omega$.

To show that $Uf_n\in{\rm dom }(P)$, we notice that  $f_n\in L^2(\mathbb{R},\mathbb{C})$ implies
$Uf_n\in L^2(\mathbb{R},\mathbb{C})$ and for $\phi \in
C_0^{\infty}(\mathbb{R})$ with $\supp (\phi)=K$ and $M=\|Uf_n\|$,
\begin{align}
|\langle Uf_n&, \phi \rangle| =\left|\int_{\mathbb{R}}(Uf_n)(t)
\phi(t)\,dt\right|
 \leq \int_{\mathbb{R}}|(Uf_n)(t) \phi(t)|\,dt\\
 &\leq \left( \int_{K}|(Uf_n)(t)|^2\,dt\right)^{1/2}\,\left(\int_{K
}|\phi(s)|^2\,ds\right)^{1/2}\notag\\
&\leq M \,  \|\phi \|  \notag\ .
\end{align}
Moreover, with our definition of the Fourier transform in
(\ref{FT}),
\begin{equation}
U(Xf_n)=-i(Uf_n)^\prime \ ,
\end{equation}
and, using (\ref{bb}), we conclude that $U(Xf_n) \in
L^2(\mathbb{R},\mathbb{C})$ which implies $(Uf_n)^\prime \in
L^2(\mathbb{R},\mathbb{C})$.

Finally, to show that  $U^{-1}f_n\in{\rm dom }(X)$, we notice that  $f_n\in L^2(\mathbb{R},\mathbb{C})$
implies $U^{-1}f_n\in L^2(\mathbb{R},\mathbb{C})$ and, letting
$\phi_n=U^{-1}f_n$,
\begin{align}
\int_{\mathbb{R}}s^2 |(U^{-1}f_n)(s)|^2\,ds&=\int_{\mathbb{R}}
|s\phi_n(s)|^2\,ds \\
 &=\int_{\mathbb{R}}
|(X\phi_n)(s)|^2\,ds =\|X\phi_n\|^2=\|U(X\phi_n)\|^2   \notag\\
&=\|\left( U\phi_n \right)^\prime\|^2=\|f_n^\prime\|^2<\infty \  .
\notag
\end{align}
\end{remark}

\begin{remark} \rm
For
\begin{equation}
g(w)=\Phi(w)=\pi^{-1/4}\,e^{-\frac{w^2}{2}} \ ,
\end{equation}
the formula for the resolvent of $P$ takes the form
\begin{equation}
R(z;P)\Phi(s)=\left\{
\begin{array}{llr}
-i
\pi^{-1/4}\int_{-\infty}^s e^{iz(s-w)-\frac{w^2}{2}} \,dw  ,&\,\, {\rm Im}\, z> 0   \\
& \\
i \pi^{-1/4}\int^{\infty}_s e^{iz(s-w)-\frac{w^2}{2}} \,dw
 , &\,\, {\rm Im}\,z< 0
\end{array}
\right. \ .
\end{equation}
\end{remark}

\begin{theorem}\label{P} For $t\in \mathbb{R}$, the vacuum characteristic function of  $P$
 is
\begin{equation}
\langle  e^{i t P}  \rangle= e^{-\frac{t^2}{4}} \ .
\end{equation}
\end{theorem}
\begin{proof}
Because of the non-analyticity of $R(z;P)\Phi(s)$  at real $z$'s,
we compute the vacuum characteristic function of $P$ as follows:
\begin{align}\label{cff}
\langle & e^{itP} \rangle=\frac{1}{2\pi i}\lim_{r\to\infty}\lim_{\varepsilon\to 0^+}\\
&\left( \oint_{C^+({\varepsilon, r})} e^{itz}\langle \Phi,
R(z;P)\Phi \rangle \,dz+\oint_{C^-({\varepsilon, r})}
e^{itz}\langle \Phi, R(z;P)\Phi \rangle \,dz\right)\ ,\notag
\end{align}
where,
\begin{equation}
C^+({\varepsilon, r})=\{z\in\mathbb{C} \,:\, |z|=r, \, \varepsilon
\leq {\rm Im }z \leq r \} \  ,
\end{equation}
and
\begin{equation}
 C^-({\varepsilon, r})=\{z\in\mathbb{C} \,:\, |z|=r, \, -r\leq {\rm Im }z\leq -\varepsilon  \}  \ ,
\end{equation}
are  semi-circular paths both  centered at $z=0$  with radius $r$.
The ends of the arcs, in the complex plane,
 are the points  $\pm\sqrt{r^2-\varepsilon^2}+i\varepsilon$ and
 $\pm\sqrt{r^2-\varepsilon^2}-i\varepsilon$, respectively.

We close $C^+({\varepsilon, r})$ and $C^-({\varepsilon, r})$ with
the linear paths
\begin{equation}\label{ap}
a^+({\varepsilon, r})=\{z\in\mathbb{C} \,:\, z=x+i\varepsilon, \,
-\sqrt{r^2-\varepsilon^2} \leq x \leq \sqrt{r^2-\varepsilon^2} \}
\ ,
\end{equation}
and
\begin{equation}
a^-({\varepsilon, r})=\{z\in\mathbb{C} \,:\, z=x-i\varepsilon, \,
-\sqrt{r^2-\varepsilon^2} \leq x \leq \sqrt{r^2-\varepsilon^2} \}\
,
\end{equation}
respectively.

 In view of (\ref{res}), $ e^{itz}\langle \Phi,
R(z;P)\Phi \rangle$ is an analytic function of $z$, inside and on
the closed contours
 \begin{equation}
\gamma^+({\varepsilon, r}):= C^+({\varepsilon,
r})+a^+({\varepsilon, r}) \,,\,\gamma^-({\varepsilon, r}):=
C^-({\varepsilon, r})+a^-({\varepsilon, r}) \  .
\end{equation}
Thus, by Cauchy's theorem,
\begin{equation}
\oint_{\gamma^+({\varepsilon, r})} e^{itz}\langle \Phi, R(z;P)\Phi
\rangle \,dz=\oint_{\gamma^-({\varepsilon, r})} e^{itz}\langle
\Phi, R(z;P)\Phi \rangle \,dz=0 \ ,
\end{equation}
which implies that
\begin{equation}
\oint_{C^+({\varepsilon, r})} e^{itz}\langle \Phi, R(z;P)\Phi
\rangle \,dz=-\oint_{a^+({\varepsilon, r})} e^{itz}\langle \Phi,
R(z;P)\Phi \rangle \,dz \ ,
\end{equation}
and
\begin{equation}
\oint_{C^-({\varepsilon, r})} e^{itz}\langle \Phi, R(z;P)\Phi
\rangle \,dz=-\oint_{a^-({\varepsilon, r})} e^{itz}\langle \Phi,
R(z;P)\Phi \rangle \,dz \  .
\end{equation}
Using (\ref{res}), (\ref{ap})  and (\ref{phi}) we have:
\begin{align}
&\oint_{a^+({\varepsilon, r})} e^{itz}\langle \Phi, R(z;P)\Phi
\rangle \,dz\\
&=\int_{-\sqrt{r^2-\varepsilon^2}}^{\sqrt{r^2-\varepsilon^2}}
e^{it(x+i\varepsilon)}\langle \Phi,
R(x+i\varepsilon;P)\Phi \rangle \,dx\notag\\
&=\int_{-\sqrt{r^2-\varepsilon^2}}^{\sqrt{r^2-\varepsilon^2}}
e^{it(x+i\varepsilon)}\int_{-\infty}^\infty \Phi(s)
R(x+i\varepsilon;P)\Phi(s)\,ds \,dx \notag\\
&=-i
\pi^{-1/2}\int_{-\sqrt{r^2-\varepsilon^2}}^{\sqrt{r^2-\varepsilon^2}}
e^{it(x+i\varepsilon)}\int_{-\infty}^\infty
e^{-\frac{s^2}{2}}\int_{-\infty}^s
e^{i(x+i\varepsilon)(s-w)-\frac{w^2}{2}}\,dw\,ds \,dx \notag \   .
\end{align}
Thus
\begin{align}
&\lim_{r\to\infty}\lim_{\varepsilon\to
0^+}\oint_{a^+({\varepsilon, r})} e^{itz}\langle
\Phi, R(z;P)\Phi \rangle \,dz\\
&=-i \pi^{-1/2}\int_{-\infty}^{\infty}
e^{itx}\int_{-\infty}^\infty e^{-\frac{s^2}{2}}\int_{-\infty}^s
e^{ix(s-w)-\frac{w^2}{2}}\,dw\,ds \,dx \notag
\\
&=-i \pi^{-1/2}\int_{-\infty}^\infty e^{itx}\int_{-\infty}^\infty
e^{-\frac{s^2}{2}+ixs}\int_{-\infty}^s
e^{-ixw-\frac{w^2}{2}}\,dw\,ds \,dx \notag \ .
\end{align}
Similarly,
\begin{align}
&\oint_{a^-({\varepsilon, r})} e^{itz}\langle \Phi, R(z;P)\Phi
\rangle \,dz\\
&=\int_{\sqrt{r^2-\varepsilon^2}}^{-\sqrt{r^2-\varepsilon^2}}
e^{it(x-i\varepsilon)}\langle \Phi,
R(x-i\varepsilon;P)\Phi \rangle \,dx\notag\\
&=-\int_{-\sqrt{r^2-\varepsilon^2}}^{\sqrt{r^2-\varepsilon^2}}
e^{it(x-i\varepsilon)}\langle \Phi,
R(x-i\varepsilon;P)\Phi \rangle \,dx\notag\\
&=-\int_{-\sqrt{r^2-\varepsilon^2}}^{\sqrt{r^2-\varepsilon^2}}
e^{it(x-i\varepsilon)}\int_{-\infty}^\infty \Phi(s)
R(x-i\varepsilon;P)\Phi(s)\,ds \,dx \notag\\
&=-i
\pi^{-1/2}\int_{-\sqrt{r^2-\varepsilon^2}}^{\sqrt{r^2-\varepsilon^2}}
e^{it(x-i\varepsilon)}\int_{-\infty}^\infty
e^{-\frac{s^2}{2}}\int_s^{\infty}
e^{i(x-i\varepsilon)(s-w)-\frac{w^2}{2}}\,dw\,ds \,dx \notag \   .
\end{align}
Thus
\begin{align}
&\lim_{r\to\infty}\lim_{\varepsilon\to
0^+}\oint_{a^-({\varepsilon, r})} e^{itz}\langle
\Phi, R(z;P)\Phi \rangle \,dz\\
&=-i \pi^{-1/2}\int_{-\infty}^{\infty}
e^{itx}\int_{-\infty}^\infty e^{-\frac{s^2}{2}}\int_{-\infty}^s
e^{ix(s-w)-\frac{w^2}{2}}\,dw\,ds \,dx \notag
\\
&=-i \pi^{-1/2}\int_{-\infty}^\infty e^{itx}\int_{-\infty}^\infty
e^{-\frac{s^2}{2}+ixs}\int_{-\infty}^s
e^{-ixw-\frac{w^2}{2}}\,dw\,ds \,dx \notag \ .
\end{align}
Therefore, by (\ref{cff})
\begin{align}
&\langle  e^{itP} \rangle=\frac{1}{2\pi i}\lim_{r\to\infty}\lim_{\varepsilon\to 0^+}\\
&\left( \oint_{C^+({\varepsilon, r})} e^{itz}\langle \Phi,
R(z;P)\Phi \rangle \,dz+\oint_{C^-({\varepsilon, r})}
e^{itz}\langle \Phi,
R(z;P)\Phi \rangle \,dz\right)\notag\\
&=-\frac{1}{2\pi i}\lim_{r\to\infty}\lim_{\varepsilon\to 0^+}\notag\\
&\left( \oint_{a^+({\varepsilon, r})} e^{itz}\langle \Phi,
R(z;P)\Phi \rangle \,dz+\oint_{a^-({\varepsilon, r})}
e^{itz}\langle \Phi,
R(z;P)\Phi \rangle \,dz\right)\notag\\
&=-\frac{1}{2\pi i}\left(-i \pi^{-1/2}\int_{-\infty}^\infty
e^{itx}\int_{-\infty}^\infty
e^{-\frac{s^2}{2}+ixs}\int_{-\infty}^s
e^{-ixw-\frac{w^2}{2}}\,dw\,ds \,dx \right. \notag\\
&\left.-i \pi^{-1/2}\int_{-\infty}^\infty
e^{itx}\int_{-\infty}^\infty
e^{-\frac{s^2}{2}+ixs}\int_{-\infty}^s
e^{-ixw-\frac{w^2}{2}}\,dw\,ds \,dx \notag \right)\notag\\
&=\frac{1}{2\pi^{3/2}}\int_{-\infty}^\infty
e^{itx}\int_{-\infty}^\infty
e^{-\frac{s^2}{2}+ixs}\int_{-\infty}^\infty
e^{-ixw-\frac{w^2}{2}}\,dw\,ds \,dx  \notag \  ,
\end{align}
which, using  (see (\ref{i1})) the formula
\begin{equation}
\int_{-\infty}^\infty
e^{-ixw-\frac{w^2}{2}}\,dw=\int_{-\infty}^\infty
e^{ixs-\frac{s^2}{2}}\,ds=\sqrt{2\pi}e^{-\frac{x^2}{2}} \  ,
\end{equation}
becomes,
\begin{align}
&\langle  e^{itP} \rangle=\frac{1}{\pi
\sqrt{2}}\int_{-\infty}^\infty
e^{itx-\frac{x^2}{2}}\int_{-\infty}^\infty
e^{ixs-\frac{s^2}{2}}\,ds \,dx \\
&=\frac{1}{ \sqrt{\pi}}\int_{-\infty}^\infty
e^{itx-\frac{x^2}{2}}e^{-\frac{x^2}{2}} \,dx =\frac{1}{
\sqrt{\pi}}\int_{-\infty}^\infty
e^{itx-x^2} \,dx \notag\\
&=\frac{1}{ \sqrt{\pi}}
\sqrt{\pi}e^{-\frac{t^2}{4}}=e^{-\frac{t^2}{4}}\notag\  .
\end{align}
\end{proof}

\begin{remark} \rm

In Propositions 9.1 and 9.2 of  \cite{BFspec}, the calculation at
the end of their proofs should read,
\begin{equation}
\int_{\mathbb{R}}e^{it\lambda}\,d\langle\Phi, E_{\lambda}\Phi
\rangle =\int_{\mathbb{R}}e^{it\lambda}\,\frac{1}{\sqrt{\pi}}
e^{-\lambda^2 } \,d\lambda
=\frac{1}{\sqrt{\pi}}\int_{\mathbb{R}}e^{it\lambda-\lambda^2}
 \,d\lambda
= e^{-\frac{t^2}{4}} \   ,
\end{equation}
which is the characteristic function of both $X$ and $P$.

\end{remark}

\section{The Vacuum Characteristic Function of $X+P$ }

\begin{theorem}\label{sa} The operator $X+P$ is essentially self-adjoint on  $\Omega$.
\end{theorem}

\begin{proof} Let $T=X+P$ and let $f,g
\in \Omega$.  Since $X, P$ are self-adjoint
\begin{align*}
\langle Tf, g\rangle=& \langle Xf, g \rangle+ \langle Pf, g
\rangle =\langle f, Xg \rangle+ \langle f, Pg \rangle\\
=&\langle f, Xg \rangle+ \langle f, Pg \rangle =\langle f, (X+P)g
\rangle =\langle f, Tg \rangle \ ,
\end{align*}
so $T$ is symmetric, i.e., $T\subseteq T^* $ on $\Omega$. To show
that $T$ is essentially self-adjoint, it suffices (see
\cite{Richtmyer} p.157) to show that the equation
\begin{equation}\label{vn}
(X+P)f\pm i f=g \  ,
\end{equation}
has a solution $f\in \Omega$ for each $g$ in a dense subset of
$L^2(\mathbb{R},\mathbb{C})$. For $g \in
C_0^{\infty}(\mathbb{R})$, using the formula for the resolvent of
$X+P$ obtained in Theorem 9.3 of \cite{BFspec}, see also formula
(\ref{rr}) below, we find that
\begin{equation}\label{rrr}
f(s)=-R(z;X+P)g(s)=\left\{
\begin{array}{llr}
 i\int_{-\infty}^s e^{i\frac{t^2-s^2}{2}}e^{t-s}
g(t) \,dt , &\,\,\mbox{ for  } z=i \\
& \\
-i\int^{\infty}_s e^{i\frac{t^2-s^2}{2}}e^{s-t}g(t) \,dt
,&\,\,\mbox{ for  } z=-i
\end{array}
\right.  \ .
\end{equation}
For $z=i$ we can show that $f\in \Omega$ as follows. Let $a=\min
\supp (g)$ and $b=\max \supp (g)$. Then
\begin{equation}
 f(s)=
 i\int_{(-\infty, s]\cap [a, b]} e^{i\frac{t^2-s^2}{2}}e^{t-s}
g(t) \,dt \  ,
\end{equation}
and
\begin{equation}\label{de}
 \int_\mathbb{R} |f(s)|^2\,ds=
 \int_{(-\infty, a)} |f(s)|^2 \,ds + \int_{[a, b]} |f(s)|^2 \,ds+ \int_{(b, \infty )} |f(s)|^2 \,ds \
 .
\end{equation}
The first integral on the right hand side of (\ref{de}) is equal
to zero, while for the second integral, using the
Cauchy--Bunyakovsky--Schwarz inequality, we have
\begin{align}
 \int_{[a, b]} |f(s)|^2 \,ds=&\int_a^b \left|\int_a^s e^{i\frac{t^2-s^2}{2}}e^{t-s}
g(t) \,dt  \right|^2 \, ds \leq \int_a^b \left( \int_a^s e^{t-s}
|g(t)| \,dt   \right)^2\,ds \\
 \leq & \int_a^b \left( \int_a^s e^{2(t-s)}\,dt   \right) \left(
\int_a^s|g(t)|^2 \,dt   \right)\,ds  \notag   \\
\leq&\|g\|^2 \int_a^b \int_a^s e^{2(t-s)}\,dt\,ds \leq \|g\|^2
\int_a^b \int_a^s 1\,dt\,ds \notag \\
=&\frac{1}{2} (b-a)^2 \|g\|^2<\infty \notag  \ .
\end{align}
Similarly, for the third integral on the right hand side of
(\ref{de}),
\begin{align}
 \int_{(b, \infty)} |f(s)|^2 \,ds
\leq&\|g\|^2 \int_b^{\infty} \int_a^b e^{2(t-s)}\,dt\,ds \\
\leq&\frac{1}{4} \left( 1-e^{2(a-b)} \right)\|g\|^2<\infty \notag
\ .
\end{align}
Thus $f\in L^2(\mathbb{R},\mathbb{C})$. To show that
$\int_{\mathbb{R}}s^2\,|f(s)|^2\,ds<\infty$, the corresponding
bounds are
\begin{align}
\int_{[a, b]} s^2 |f(s)|^2 \,ds \leq &\left(
\frac{1}{24}(4b^3-4a^3-6a^2-6a-3)\right.
\\
&\left.+\frac{1}{8}(1+2b+2b^2)\left( 1-e^{2(a-b)} \right)
\right)\|g\|^2
<\infty  \notag\ ,\\
 \int_{(b, \infty)}s^2 |f(s)|^2 \,ds
\leq&\frac{1}{8}(1+2b+2b^2) \left( 1-e^{2(a-b)}
\right)\|g\|^2<\infty \ .
\end{align}
To show that $\int_{\mathbb{R}}|f^\prime(s)|^2\,ds<\infty$, where
\begin{equation}
 f^\prime(s)=i g(s)+(s-1)
 \int_{-\infty}^s e^{i\frac{t^2-s^2}{2}}e^{t-s}
g(t) \,dt \  ,
\end{equation}
with
\begin{equation}
\int_{\mathbb{R}}|f^\prime(s)|^2\,ds \leq 2 \|g\|^2 +2
\int_{\mathbb{R}}|G(s)|^2\,ds \  ,
\end{equation}
where,
\begin{equation}
G(s)=\int_{-\infty}^s (s-1)e^{i\frac{t^2-s^2}{2}}e^{t-s} g(t) \,dt
\ ,
\end{equation}
using the decomposition (\ref{de}) for $G$,
 the corresponding bounds are
\begin{align}
&\int_{(-\infty, a)} |G(s)|^2 \,ds = 0 \ ,\\
\int_{[a, b]} |G(s)|^2 \,ds \leq &\|g\|^2 \int_a^b \int_a^s (s-1)^2e^{2(t-s)}\,dt\,ds  \\
  = & \left(
\frac{1}{24}(4b^3-12b^2+12b -4a^3+6a^2-6a-3)\right. \notag\\
 &\left.+\frac{1}{8}(1-2b+2b^2)\left( 1-e^{2(a-b)} \right)
\right)\|g\|^2
<\infty \notag  \ ,\\
 \int_{(b, \infty)}|G(s)|^2 \,ds
\leq& \frac{1}{8}(1-2b+2b^2)\left( 1-e^{2(a-b)}\right) \|g\|^2
<\infty \ .
\end{align}
The absolute continuity of $f$ follows from the fact that
$f^{\prime}$ is in $L^2(\mathbb{R},\mathbb{C})$ as a distribution,
since, for all $\phi \in C_0^{\infty}(\mathbb{R})$, by the
Cauchy--Bunyakovsky--Schwarz inequality,
\begin{equation}
|\langle f^{\prime}, \phi \rangle|
=\left|\int_{\mathbb{R}}f^{\prime}(s) \phi(s)\,ds\right|
 \leq
 M\, \|\phi \| \,,\, M= \| f^\prime\|<\infty   \ .
\end{equation}

 To show that the solution of (\ref{vn}) is unique, it
suffices to show that the equation
\begin{equation}\label{vn2}
(X+P)f\pm i f=0\,,\,f\in \Omega \  ,
\end{equation}
has  the unique solution $f=0$. Using the definition of $X$ and
$P$,  (\ref{vn2}) becomes
\begin{equation}\label{vn3}
(s\pm  i)f(s)-i f^\prime(s)=0 \,,\, s\in\mathbb{R}\  ,
\end{equation}
which, letting $a={\rm Re }f$ and $b={\rm Im }f $, gives the
system
\begin{align}
s \,a(s)\mp b(s)+b^\prime (s)&=0\label{ab}\ ,\\
s \,b(s)\pm a(s)-a^\prime (s)&=0\label{ba} \ .
\end{align}
Multiplying  (\ref{ab}) and (\ref{ba}) by $b(s)$ and $a(s)$,
respectively, and subtracting the resulting equations, we obtain
\begin{equation}\label{vn4}
\left(a(s)^2+b(s)^2\right)^\prime= \pm 2
\left(a(s)^2+b(s)^2\right)\  ,
\end{equation}
which implies that
\begin{equation}
 |f(s)|^2=a(s)^2+b(s)^2=c e^{\pm 2 s} \, , \, c\in\mathbb{R}\  ,
\end{equation}
which vanishes at $\pm \infty $ if and only if $c=0$. Thus $f=0$.
The proof for $z=-i$ is similar.
\end{proof}

\begin{theorem}\label{X+Pna}  The operator  $X+P$ has only continuous spectrum consisting
 of the entire
real line. Moreover, for $g$ in the range of $z-(X+P)$, the
resolvent operator $R(z; X+P)$ is
\begin{equation}\label{rr}
R(z;X+P)g(s)=\left\{
\begin{array}{llr}
 -i\int_{-\infty}^s e^{i\frac{t^2-s^2}{2}}e^{-iz(t-s)}
g(t) \,dt , &\,\, {\rm Im}\, z> 0 \\
& \\
i\int^{\infty}_s e^{i\frac{t^2-s^2}{2}}e^{-iz(t-s)}g(t) \,dt
,&\,\, {\rm Im}\, z<0
\end{array}
\right. \ .
\end{equation}
\end{theorem}
\begin{proof} Formula (\ref{rr}) was proved in \cite{BFspec}. Since, by Theorem \ref{sa},
 $X+P$ is essentially self-adjoint, its spectrum is real. If $z\in\mathbb{R}$ is an
 eigenvalue of $X+P$, then a corresponding
eigenfunction $G$ would satisfy
\begin{equation}
i\,G^{\prime}(s)+(z-s)G(s)=0 \,,\, s\in\mathbb{R}\  ,
\end{equation}
which would imply that
\begin{equation}
G(s)=c e^{izs}e^{-i\frac{s^2}{2}}\, , c\in\mathbb{C}\  ,
\end{equation}
which is in $L^2(\mathbb{R},\mathbb{C})$ if and only if $c=0$,
i.e., if and only if $G=0$. Thus $X+P$ has no point spectrum.

To show that every real number $z$ is in the continuous spectrum
of $X+P$, we will show that the range of $z-(X+P)$ is dense in
$L^2(\mathbb{R},\mathbb{C})$ for every $z\in\mathbb{R}$ and that
for every  $\varepsilon >0$  there exists an \textit{approximate
eigenvector}  $g_\varepsilon$  of $z-(X+P)$  (see \cite{Richtmyer}
p.144).

For real $z$, let $ \phi\in \mathcal{S}$ and consider the unitary
transformation
\begin{equation}
V: L^2(\mathbb{R},\mathbb{C}) \to L^2(\mathbb{R},\mathbb{C})\,,\,
(V\psi)(s)=e^{\frac{is^2}{2}}\psi(s)\,,\,(V^*\psi)(s)=e^{-\frac{is^2}{2}}\psi(s)\
,
\end{equation}
which satisfies (the proof is along the lines of Remark \ref{rem})
\begin{equation}
V^*PV\psi=(X+P)\psi\,,\,\psi\in\mathcal{S} \ .
\end{equation}
 Let $g=V\phi$, let $g_n$ be as in (\ref{tb}) with
$(z-P)g_n\to g$, and let $\phi_n=V^*g_n$. Then  $(z-P)V\phi_n\to
V\phi$ implies $V^*(z-P)V\phi_n\to \phi$, i.e., $(z-(X+P))\phi_n
\to \phi$ in the $L^2$-sense. Thus,  the range of $z-(X+P)$ is
dense in $\mathcal{S}$ therefore dense in
$L^2(\mathbb{R},\mathbb{C})$.

For an arbitrary $\varepsilon >0$, we consider  the function
\begin{equation}
g_\varepsilon(s)=\varepsilon^{1/2} \pi^{-1/4}
 e^{-\frac{\varepsilon^2s^2 }{2}}e^{i ( sz-\frac{s^2}{2}) }\  ,
\end{equation}
with
\begin{equation}
g_\varepsilon^{\prime}(s)=\varepsilon^{1/2}
\pi^{-1/4}(-s\varepsilon^2+i(z-s))
 e^{-\frac{\varepsilon^2s^2 }{2}}e^{i ( sz-\frac{s^2}{2}) }  \ .
\end{equation}
 Then,
\begin{align}
\|g_\varepsilon\|&=1 \ ,\\
\int_{\mathbb{R}} s^2 |g_\varepsilon(s)|^2 \,ds&=\frac{1}{2\varepsilon^2}< \infty \  ,\\
\int_{\mathbb{R}}
|g_\varepsilon^{\prime}(s)|^2\,ds&=\frac{1}{2\varepsilon^2}+\frac{\varepsilon^2}{2}+z^2
< \infty\  ,
\end{align}
so $g_\varepsilon$ is a unit vector in the domain of $X$ and $P$.
The absolute continuity of $g_\varepsilon$ follows from the fact
that $g_\varepsilon^{\prime}$ is in $L^2(\mathbb{R},\mathbb{C})$,
as a distribution, since, for all $\phi \in
C_0^{\infty}(\mathbb{R})$,
\begin{align}
&|\langle g_\varepsilon^{\prime}, \phi \rangle| =\left|\int_{\mathbb{R}}g_\varepsilon^{\prime}(s) \phi(s)\,ds\right|\\
\notag
&=\left|\int_{\mathbb{R}}\left(\varepsilon^{1/2}
\pi^{-1/4}(-s\varepsilon^2+i(z-s))
 e^{-\frac{\varepsilon^2s^2 }{2}}e^{i ( sz-\frac{s^2}{2}) }  \right) \phi(s)\,ds\right|
 \\\notag
 &\leq \varepsilon^{1/2}
\pi^{-1/4} \int_{\mathbb{R}} \sqrt{s^2 \varepsilon^4+(z-s)^2}
e^{-\frac{\varepsilon^2s^2 }{2}} |\phi(s)| \,ds \\ \notag
 &\leq
\varepsilon^{1/2} \pi^{-1/4} \left(\int_{\mathbb{R}} \left(s^2
\varepsilon^4+(z-s)^2\right) e^{-\varepsilon^2 s^2 }  \,ds
\right)^{1/2} \, \left(\int_{\mathbb{R}} |\phi(s)|^2\,ds
\right)^{1/2} \\\notag
 & = \varepsilon^{1/2} \pi^{-1/4}  \left( \frac{ \pi^{1/2}  (1+2 z^2 \varepsilon^2+\varepsilon^4  }{2 \varepsilon^3  }\right)^{1/2}
 \|\phi\| \\\notag
 &=\left( \frac{1}{2 \varepsilon^2}+\frac{ \varepsilon^2 }{2}+z^2 \right)^{1/2} \|\phi\| \
 .
\end{align}
 Moreover,
\begin{equation}
z g_\varepsilon(s) -s g_\varepsilon(s) +i
g_\varepsilon^{\prime}(s)=-i \varepsilon^{5/2} \pi^{-1/4}
e^{isz}e^{-\frac{i s^2 }{2}} e^{-\frac{\varepsilon^{2} s^2 }{2} }\
,
\end{equation}
implies
\begin{equation}
\int_{\mathbb{R}} |z g_\varepsilon(s) -s g_\varepsilon(s) +i
g_\varepsilon^{\prime}(s)|^2\,ds=\varepsilon^{5}
\pi^{-1/2}\int_{\mathbb{R}} s^2  e^{-\varepsilon^{2} s^2 } \,ds
=\frac{\varepsilon^2}{2} < \varepsilon^2\  ,
\end{equation}
which means that
\begin{equation}
\|\left(z-(X+P)\right)g_\varepsilon\| < \varepsilon\  ,
\end{equation}
i.e., $g_\varepsilon$ is an \textit{approximate eigenvector} of
$X+P$, so (see \cite{Richtmyer} p.144)  $z$ is in the continuous
spectrum of $X+P$.
\end{proof}

\begin{theorem}\label{X+P} For $t\in \mathbb{R}$, the vacuum characteristic function
 of  $X+P$ is
\begin{equation}
\langle  e^{i t (X+P)}  \rangle=
 e^{-\frac{t^2}{2}} \ .
\end{equation}
\end{theorem}
\begin{proof} As in the proof of Theorem \ref{P}, using formula
(\ref{rr}) for the resolvent of $T:=X+P$ and integration formula
(\ref{i1}), we have
\begin{align}
&\langle  e^{itT} \rangle=
-\frac{1}{2\pi i}\lim_{r\to\infty}\lim_{\varepsilon\to 0^+}\\
&\left( \oint_{a^+({\varepsilon, r})} e^{itz}\langle \Phi,
R(z;T)\Phi \rangle \,dz+\oint_{a^-({\varepsilon, r})}
e^{itz}\langle \Phi,
R(z;T)\Phi \rangle \,dz\right)\notag\\
&=-\frac{1}{2\pi i}\lim_{r\to\infty}\lim_{\varepsilon\to
0^+}\left(\int_{-\sqrt{r^2-\varepsilon^2}}^{\sqrt{r^2-\varepsilon^2}}
e^{it(x+i\varepsilon)}\langle \Phi,
R(x+i\varepsilon;T)\Phi \rangle \,dx \right.\notag\\
&\left.+\int_{\sqrt{r^2-\varepsilon^2}}^{-\sqrt{r^2-\varepsilon^2}}
e^{it(x-i\varepsilon)}\langle \Phi,
R(x-i\varepsilon;T)\Phi \rangle \,dx \right)\notag\\
&=-\frac{1}{2\pi i}\lim_{r\to\infty}\lim_{\varepsilon\to
0^+}\left(\int_{-\sqrt{r^2-\varepsilon^2}}^{\sqrt{r^2-\varepsilon^2}}
e^{it(x+i\varepsilon)}\int_{-\infty}^\infty\Phi(s)
R(x+i\varepsilon;T)\Phi(s)\,ds \,dx \right.\notag\\
&\left.+\int_{\sqrt{r^2-\varepsilon^2}}^{-\sqrt{r^2-\varepsilon^2}}
e^{it(x-i\varepsilon)} \int_{-\infty}^\infty\Phi(s)
R(x-i\varepsilon;T)\Phi(s)\,ds \,dx     \right)\notag\\
&=-\frac{1}{2\pi i}\left(-i
\pi^{-1/2}\right)\lim_{r\to\infty}\lim_{\varepsilon\to
0^+}\notag\\
&\left(\int_{-\sqrt{r^2-\varepsilon^2}}^{\sqrt{r^2-\varepsilon^2}}
e^{it(x+i\varepsilon)}\int_{-\infty}^\infty
e^{-\frac{s^2}{2}}\int_{-\infty}^s
e^{i\frac{w^2-s^2}{2}}e^{-i(x+i\varepsilon)(w-s)}e^{-\frac{w^2}{2}}\,
dw \,ds \,dx \right.\notag\\
&\left.
-\int_{\sqrt{r^2-\varepsilon^2}}^{-\sqrt{r^2-\varepsilon^2}}
e^{it(x-i\varepsilon)}\int_{-\infty}^\infty
e^{-\frac{s^2}{2}}\int_s^{\infty}
e^{i\frac{w^2-s^2}{2}}e^{-i(x-i\varepsilon)(w-s)}e^{-\frac{w^2}{2}}\,
dw \,ds \,dx \right)\notag\\
&=-\frac{1}{2\pi i}\left(-i
\pi^{-1/2}\right)\lim_{r\to\infty}\lim_{\varepsilon\to
0^+}\notag\\
&\left(\int_{-\sqrt{r^2-\varepsilon^2}}^{\sqrt{r^2-\varepsilon^2}}
e^{it(x+i\varepsilon)}\int_{-\infty}^\infty
e^{-\frac{s^2}{2}}\int_{-\infty}^s
e^{i\frac{w^2-s^2}{2}}e^{-i(x+i\varepsilon)(w-s)}e^{-\frac{w^2}{2}}\,
dw \,ds \,dx \right.\notag\\
&\left.
+\int_{-\sqrt{r^2-\varepsilon^2}}^{\sqrt{r^2-\varepsilon^2}}
e^{it(x-i\varepsilon)}\int_{-\infty}^\infty
e^{-\frac{s^2}{2}}\int_s^{\infty}
e^{i\frac{w^2-s^2}{2}}e^{-i(x-i\varepsilon)(w-s)}e^{-\frac{w^2}{2}}\,
dw \,ds \,dx \right)\notag\\
&=\frac{1}{2\pi^{3/2}}\left(\int_{-\infty}^{\infty}
e^{itx}\int_{-\infty}^\infty e^{-\frac{s^2}{2}}\int_{-\infty}^s
e^{i\frac{w^2-s^2}{2}}e^{-ix(w-s)}e^{-\frac{w^2}{2}}\,
dw \,ds \,dx \right.\notag\\
&\left. +\int_{-\infty}^{\infty} e^{itx}\int_{-\infty}^\infty
e^{-\frac{s^2}{2}}\int_s^{\infty}
e^{i\frac{w^2-s^2}{2}}e^{-ix(w-s)}e^{-\frac{w^2}{2}}\,
dw \,ds \,dx \right)\notag\\
&=\frac{1}{2\pi^{3/2}}\int_{-\infty}^{\infty}
e^{itx}\int_{-\infty}^\infty
e^{-\frac{s^2}{2}}\int_{-\infty}^\infty
e^{i\frac{w^2-s^2}{2}}e^{-ix(w-s)}e^{-\frac{w^2}{2}}\,
dw \,ds \,dx \notag\\
&=\frac{1}{2\pi^{3/2}} \int_{-\infty}^{\infty}
e^{itx}\left(\int_{-\infty}^\infty
e^{-\frac{1+i}{2}s^2}e^{ixs}\,ds\right)\left(
\int_{-\infty}^\infty e^{-\frac{1-i}{2}w^2}e^{-ixw}\,dw
\right)\,dx\notag\\
&=\frac{1}{2\pi^{3/2}} \int_{-\infty}^{\infty}
e^{itx}\left(\left(\frac{2\pi}{1+i}\right)^{1/2} e^{
-\frac{1-i}{4} x^2 }
\right)\left(\left(\frac{2\pi}{1-i}\right)^{1/2}e^{
-\frac{1+i}{4} x^2 }\right)\,dx\notag\\
&=\frac{1}{\sqrt{2\pi}} \int_{-\infty}^{\infty}e^{itx}e^{
-\frac{x^2}{2}  }\,dx=\frac{1}{\sqrt{2\pi}}\sqrt{2\pi}e^{
-\frac{t^2}{2}  }=e^{ -\frac{t^2}{2}  }\notag \ .
\end{align}
\end{proof}

\section{The Vacuum Characteristic Function of $XP+PX$}

\begin{lemma}\label{den}  The subset $D$ of
$C_{0}^{\infty}(\mathbb{R})$ consisting of bump functions
vanishing in a neighborhood of zero is  dense in
$L^2(\mathbb{R},\mathbb{C})$.
\end{lemma}

\begin{proof} Let $f \in L^2(\mathbb{R},\mathbb{C})$ and let $\delta >0$ be given.
Let $\phi \in
C_{0}^{\infty}(\mathbb{R})$ be such that $\|f-\phi \|<\delta /3$.
Such a $\phi$ exists since  $ C_{0}^{\infty}(\mathbb{R})$ is dense
in $L^2(\mathbb{R},\mathbb{C})$. Let $\varepsilon >0$ be such that
$\|\phi -\phi \chi_{ {}_{E}}\|<\delta/3$, where
$E=\mathbb{R}-(-\varepsilon , \varepsilon )$ and $\chi_{ {}_{E}}$
is the characteristic (indicator) function of $E$. Such an
$\varepsilon>0$ exists because
\begin{equation}
\int_{\mathbb{R}}|\phi(x)-(\phi \chi_{
{}_{E}})(x)|^2\,dx=\int_{(-\varepsilon , \varepsilon
)}|\phi(x)|^2\,dx \to 0\,,\,\varepsilon\to 0^+\  ,
\end{equation}
since the Lebesgue measure of $(-\varepsilon , \varepsilon )$ goes
to zero as $\varepsilon\to 0^+$ (\cite{friedman}, Corollary
2.8.5). Finally, let $\psi_{\eta}$, $\eta>0$,  be a
\textit{mollifier}  such that, $\|\phi \chi_{ {}_{E}}- \phi \chi_{
{}_{E}}*\psi_{\eta}\|<\delta /3$. Such an $\eta>0$ exists since
$\|\phi \chi_{ {}_{E}}- \phi \chi_{ {}_{E}}*\psi_{\eta}\|\to 0$ as
$\eta\to 0^+$ (\cite{Richtmyer} p.93, \cite{folland} Theorem 7.4
p.210)). Then, by the triangle inequality, $\|f-\phi \chi_{
{}_{E}}*\psi_{\eta}\|<\delta$.

The function $\phi \chi_{ {}_{E}}*\psi_{\eta}$ vanishes in
$(-\varepsilon , \varepsilon )$ and it is infinitely
differentiable (since convolution inherits the smoothness
properties of $\psi_{\eta}$). To show that $\phi \chi_{
{}_{E}}*\psi_{\eta}$ is compactly supported,  since $\psi_{\eta}$
is compactly supported , it suffices to show that the support of
 $\phi \chi_{ {}_{E}}$ is compact. Let $K=\supp(\phi \chi_{ {}_{E}})$. Then,
\begin{equation}
K=\overline{ \{x\in \mathbb{R}\,:\, (\phi \chi_{ {}_{E}})(x)\neq 0
\}}=\overline{ \{x\in E\,:\,\phi(x)\neq 0 \}} \subseteq
\supp(\phi) \ .
\end{equation}
Since $\supp(\phi)$ is bounded, $K$ is also bounded, and being a
closure it is closed. Thus $K$ is compact.
\end{proof}

\begin{theorem}\label{sa2} The operator $XP+PX$ is essentially self-adjoint on
 $ \mathcal{S}$.
\end{theorem}

\begin{proof} In Theorem 2.1 of \cite{BFspec2} we showed that
$XP+PX$ admits a self-adjoint extension. To show that the
self-adjoint extension is unique, i.e. that  $XP+PX$ is
essentially self-adjoint, as in the proof of Theorem \ref{sa}, it
suffices to show that $z=\pm i$ is in the resolvent of $XP+PX$,
i.e. that for $z=\pm i$ the equation
\begin{equation}\label{vn22}
(XP+PX)f+z f=g  \ ,
\end{equation}
has a (unique) solution $f\in  \mathcal{S}$ for each $g \in
C_0^{\infty}(\mathbb{R})$ vanishing in a neighborhood
$(-\varepsilon, \varepsilon)$, $\varepsilon >0$,  of $0$. As shown
in Lemma \ref{den}, the set of such $g$'s is dense in
$L^2(\mathbb{R},\mathbb{C})$. For $z=-i$, using the formula for
the resolvent $R(a;XP+PX)g(s)$ for ${\rm Im }a
>-1$ obtained in Theorem 2.2 of \cite{BFspec2}, we find that
\begin{equation}\label{rrrr}
f(s)=-R(i;XP+PX)g(s)=\left\{
\begin{array}{llr}
 -\frac{i}{2}s^{ -1 }\int_s^\infty g(w)\,dw  \   ,&  s> 0   \\
& \\
\frac{i}{2}g(0) \   ,&  s= 0 \\
& \\
\frac{i}{2}s^{ -1 }\int_{-\infty}^s g(w)\,dw\ , & s<0
\end{array}
\right. \  .
\end{equation}
For  $z=i$ the formula for the resolvent $R(a;X+P)g(s)$  in
Theorem 2.2 of \cite{BFspec2} cannot be used as in (\ref{rrrr}) to
compute $f(s)=-R(-i;XP+PX)g(s)$ since $a=-i$ does not satisfy
${\rm Im }a >-1$. Instead, we notice that for  $z=i$, using the
definition of $X$ and $P$, (\ref{vn2}) reduces to
\begin{equation}\label{wtf}
s f^{\prime}(s)=\frac{i}{2} g(s) \ .
\end{equation}
Since $f\in  \mathcal{S}$ implies that $f$ vanishes at $\pm
\infty$, integrating (\ref{wtf}) from $s$ to $+\infty$, for $s>0$,
and from
 $-\infty$ to  $s$, for $s<0$, we find that
\begin{equation}\label{rrrrr}
f(s)=\left\{
\begin{array}{llr}
 -\frac{i}{2}\int_s^\infty \frac{g(w)}{w}\,dw  \   ,&  s> 0   \\
& \\
\frac{i}{2}\int_{-\infty}^s \frac{g(w)}{w}\,dw\ , & s<0
\end{array}
\right. \  .
\end{equation}
For $s=0$, since $g(0)=  0$,  (\ref{wtf}) is satisfied by any
$f\in  \mathcal{S}$.  The proof  that $f\in \mathcal{S}$ follows
from L'Hopital's rule and the fact that $g$ vanishes near and far
from $0$, so that the integrals in (\ref{rrrr}) and (\ref{rrrrr})
vanish for "very positive" and "very negative" values of $s$. We
emphasize the importance of considering $g\in D$, where $D$ is as
in Lemma \ref{den}.  The fact that such an $f$ is unique, proceeds
along the lines of the proof of Theorem \ref{sa} and we omit the
details.
\end{proof}

\begin{theorem}\label{csp}  The operator  $XP+PX$ has only continuous spectrum consisting
of the entire
real line. Moreover, for $g$ in the range of $z-(XP+PX)$, the
resolvent operator $R(z; XP+PX)$, where ${\rm Im} z>-1$, is
\begin{equation}\label{rops}
R(z;XP+PX)g(s)=\left\{
\begin{array}{llr}
 \frac{i}{2}s^{ -\frac{z+i}{2i} }\int_s^\infty w^{ \frac{z-i}{2i}
}g(w)\,dw  \   ,&  s> 0   \\
& \\
 \frac{g(0)}{z+i}\   ,&  s= 0   \\
& \\
\frac{i}{2}(-s)^{ -\frac{z+i}{2i} }\int_{-\infty}^s (-w)^{
\frac{z-i}{2i} }g(w)\,dw\ , & s<0
\end{array}
\right. \  .
\end{equation}
\end{theorem}
\begin{proof} The proof of (\ref{rops}) can be found in Theorem 2.2 of
\cite{BFspec2}. Since, by Theorem \ref{sa2},  $XP+PX$ is
essentially self-adjoint, its spectrum is real. If
$z\in\mathbb{R}$ is an eigenvalue of $XP+PX$, then a corresponding
eigenfunction $G$ would satisfy
\begin{equation}\label{hwg}
 2 i s\,G^{\prime}(s)+(i+z)G(s)=0 \,,\, s\in\mathbb{R} \ .
\end{equation}
For $s=0$, (\ref{hwg}) implies  that $G(0)=0$. For $s\neq 0$ it
becomes
\begin{equation}\label{hwg2}
 G^{\prime}(s)+\frac{i+z}{2is}G(s)=0 \  ,
\end{equation}
which was solved in Theorem 2.2 of \cite{BFspec2} (for $g=0$ and
$a=z$  in their notation) and has the unique solution $G=0$. Thus
$XP+PX$ has no point spectrum.

To show that every real number $z$ is in the continuous spectrum
of $XP+PX$, we will show that the range of $z-(XP+PX)$ is dense in
$L^2(\mathbb{R},\mathbb{C})$ for every $z\in\mathbb{R}$ and that
for every  $\varepsilon >0$  there exists an \textit{approximate
eigenvector}  $g_\varepsilon$  of $z-(XP+PX)$  (see
\cite{Richtmyer} p.144).

For $g\in D$, where $D$ is as in Lemma \ref{den}, i.e. for $g \in
C_0^{\infty}(\mathbb{R})$ vanishing in a neighborhood
$(-\varepsilon, \varepsilon)$  of $0$, the equation
\begin{equation}\label{hwg3}
(z-(XP+PX))G=g \iff   2is\,  G^{\prime}(s)+(i+z)G(s)=g(s)\  ,
\end{equation}
has, see Theorem 2.2 of \cite{BFspec2}, for every
$z\in\mathbb{R}$, the solution
\begin{equation}\label{rops2}
G(s)=\left\{
\begin{array}{llr}
 \frac{i}{2}s^{ -\frac{z+i}{2i} }\int_s^\infty w^{ \frac{z-i}{2i}
}g(w)\,dw  \   ,&   s\geq \varepsilon   \\
& \\
 0\   ,&  s \in (-\varepsilon, \varepsilon )    \\
& \\
\frac{i}{2}(-s)^{ -\frac{z+i}{2i} }\int_{-\infty}^s (-w)^{
\frac{z-i}{2i} }g(w)\,dw\ , & s\leq -\varepsilon
\end{array}
\right. \  ,
\end{equation}
which is a.e. in $\mathcal{S}$. Thus the range of of $z-(XP+PX)$
is dense in $L^2(\mathbb{R},\mathbb{C})$ for every
$z\in\mathbb{R}$.

Moreover, if $z\in\mathbb{R}$ then, for every  $\varepsilon
>0$,  the function $g_\varepsilon$ defined by

\begin{equation}\label{rops3}
g_\varepsilon(s)=\left\{
\begin{array}{llr}
\frac{1}{   \sqrt{   \Gamma(0, 2 \varepsilon^2 ) }  } e^{\left(-\frac{1}{2}+i \frac{z}{2}\right) \ln |s|}e^{-s^2}\   ,&   s\notin (-\varepsilon, \varepsilon )    \\
& \\
 0\   ,&  s \in (-\varepsilon, \varepsilon )
\end{array}
\right. \  ,
\end{equation}
where
\begin{equation}
\Gamma(a, b):=\int_b^\infty t^{a-1} e^{-t} \,dt \  ,
\end{equation}
is the \textit{incomplete Gamma function}, is an
\textit{approximate eigenvector} of $z-(XP+PX)$ because it is a.e.
in $\mathcal{S}$ and satisfies
\begin{align}
\|g_\varepsilon \|&=1 \  ,\\
(z-(XP+PX))g_\varepsilon(s)&=\left\{
\begin{array}{llr}
\frac{-4 i}{   \sqrt{   \Gamma(0, 2 \varepsilon^2 ) }  } e^{\left(\frac{3}{2}+i \frac{z}{2}\right) \ln |s|}e^{-s^2}\   ,&   s\notin (-\varepsilon, \varepsilon )    \\
& \\
 0\   ,&  s \in (-\varepsilon, \varepsilon )
\end{array}
\right. \ ,\\
 \|(z-(XP+PX))g_\varepsilon
\|&=\frac{2}{ \sqrt{   \Gamma(0, 2 \varepsilon^2 ) }
}e^{-\varepsilon^2 } \sqrt{1+2 \varepsilon^2 } \to 0
\,\,,\,\,\varepsilon \to 0^+ \  ,
\end{align}
since
\begin{equation}
\lim _{ \varepsilon \to 0^+ }  \Gamma(0, 2 \varepsilon^2 )
=+\infty \  .
\end{equation}
\end{proof}
\begin{theorem}\label{XP+PX} For
$t\in \mathbb{R}$, the vacuum characteristic function of $XP+PX$
is
\begin{equation}
\langle e^{i t (XP+PX)}  \rangle=\left({\rm sech} \,2t
\right)^{1/2} \  .
\end{equation}
\end{theorem}

\begin{proof} As in Lemma 1 of \cite{BFspec2}, by the resolvent identity
\begin{equation}
R(\lambda; T)^*=R(\bar\lambda; T^*)\  ,
\end{equation}
it follows that
\begin{align}
e^{-t \varepsilon}&\langle \Phi, R(x+i\varepsilon;XP+PX)\Phi
\rangle-e^{t \varepsilon}\langle \Phi, R(x-i\varepsilon;XP+PX)\Phi
\rangle \\
=&e^{-t \varepsilon}\langle \Phi, R(x+i\varepsilon;XP+PX)\Phi
\rangle-e^{t \varepsilon}\left(\langle \Phi,
R(x+i\varepsilon;XP+PX)\Phi \rangle \right. \notag\\
&\left. + 2i \, {\rm Im }\langle \Phi, R(x-i\varepsilon;XP+PX)\Phi
\rangle \right) \notag \  ,
\end{align}
which tends to
\begin{equation}
-2i \, {\rm Im }\langle \Phi, R(x; XP+PX)\Phi \rangle \  ,
\end{equation}
as $\varepsilon\to 0^+$.

Thus, as in the proof of Theorem \ref{X+P}, using formula
(\ref{rops}) for the resolvent of $T:=XP+PX$, we have

\begin{align}
&\langle  e^{itT} \rangle=
-\frac{1}{2\pi i}\lim_{r\to\infty}\lim_{\varepsilon\to 0^+}\\
&\left( \oint_{a^+({\varepsilon, r})} e^{itz}\langle \Phi,
R(z;T)\Phi \rangle \,dz+\oint_{a^-({\varepsilon, r})}
e^{itz}\langle \Phi,
R(z;T)\Phi \rangle \,dz\right)\notag\\
&=-\frac{1}{2\pi i}\lim_{r\to\infty}\lim_{\varepsilon\to
0^+}\left(\int_{-\sqrt{r^2-\varepsilon^2}}^{\sqrt{r^2-\varepsilon^2}}
e^{it(x+i\varepsilon)}\langle \Phi,
R(x+i\varepsilon;T)\Phi \rangle \,dx \right.\notag\\
&\left.+\int_{\sqrt{r^2-\varepsilon^2}}^{-\sqrt{r^2-\varepsilon^2}}
e^{it(x-i\varepsilon)}\langle \Phi,
R(x-i\varepsilon;T)\Phi \rangle \,dx \right)\notag\\
&=-\frac{1}{2\pi i}\lim_{\varepsilon\to
0^+}\notag\\
&\int_{-\infty}^\infty e^{i t x}\left( e^{-t \varepsilon}\langle
\Phi, R(x+i\varepsilon;T)\Phi \rangle-e^{t \varepsilon}\langle
\Phi,
R(x-i\varepsilon;T)\Phi \rangle   \right)\, dx \notag\\
&=\frac{1}{\pi }\int_{-\infty}^\infty e^{i t x} \, {\rm Im
}\,\langle
\Phi, R(x; T)\Phi \rangle\, dx \notag\\
&=\frac{1}{\pi }\int_{-\infty}^\infty e^{i t x} \, {\rm Im
}\,\int_{-\infty}^\infty\Phi(s) R(x; T)\Phi(s)\,ds
\,dx \notag\\
&=\frac{1}{\pi^{3/2} }\int_{-\infty}^\infty e^{i t x} \, {\rm Im
}\, \left(\frac{i}{2} \int_{-\infty}^0
e^{-\frac{s^2}{2}}(-s)^{-\frac{1}{2}+i\frac{x}{2}}
\int_{-\infty}^s (-w)^{-\frac{1}{2}-i\frac{x}{2}}
e^{-\frac{w^2}{2}} \, dw\, ds
\right. \notag\\
&\left. +\frac{i}{2} \int^{\infty}_0
e^{-\frac{s^2}{2}}s^{-\frac{1}{2}+i\frac{x}{2}}
 \int^{\infty}_s
w^{-\frac{1}{2}-i\frac{x}{2}} e^{-\frac{w^2}{2}} \, dw\, ds
\right)\,dx \notag \   .
\end{align}
Letting $u=-s$ and $v=-w$ in the first integral, we see that it is
actually equal to the second one. Moreover, using ${\rm Im
}(iz)={\rm Re }z$, we find
\begin{align}
\langle  e^{it(T)} \rangle&=\frac{1}{\pi^{3/2}
}\int_{-\infty}^\infty e^{i t x} \, {\rm Re
}\left(\int^{\infty}_0  \int^{\infty}_s
e^{-\frac{s^2+w^2}{2}}s^{-\frac{1}{2}+i\frac{x}{2}}
w^{-\frac{1}{2}-i\frac{x}{2}} \, dw\, ds \right)\,dx  \\
&  =\frac{1}{\pi^{3/2} }\int_{-\infty}^\infty e^{i t x} \,
\left(\int^{\infty}_0  \int^{\infty}_s e^{-\frac{s^2+w^2}{2}} {\rm
Re }\left(s^{-\frac{1}{2}+i\frac{x}{2}}
w^{-\frac{1}{2}-i\frac{x}{2}}\right) \, dw\, ds \right)\,dx \notag
\ .
\end{align}
In Theorem 3 of \cite{BFspec2} it was shown that

\begin{align}
&\frac{1}{\pi^{3/2} }\int_{-\infty}^\infty e^{i t x} \,
\left(\int^{\infty}_0  \int^{\infty}_s  e^{-\frac{s^2+w^2}{2}}
{\rm Re }\left(s^{-\frac{1}{2}+i\frac{x}{2}}
w^{-\frac{1}{2}-i\frac{x}{2}}\right) \, dw\, ds \right)\,dx  \\
&=\left({\rm sech} \,2t \right)^{1/2} \notag\  .
\end{align}
\end{proof}

\begin{remark}\label{rem2} \rm Just as \textit{normal distribution} is typical
 of \textit{first-order quantum fields}, i.e. first-order
 in the \textit{white noise operators}
\begin{equation}
a:=\frac{X+iP}{\sqrt{2}}
\,\,,\,\,a^{\dagger}:=\frac{X-iP}{\sqrt{2}}\,\,,\,\, \lbrack a,
a^{\dagger}\rbrack=\mathbf{1} \  ,
\end{equation}
the \textit{generalized hyperbolic secant distribution} (see
\cite{hs} for details) with characteristic function
\begin{equation}
\left({\rm sech} \,\alpha t
\right)^{\rho}\,\,,\,\,\alpha>0\,\,,\,\,\rho>0 \  ,
\end{equation}
is typical of \textit{quadratic quantum fields}, i.e. of
second-order in $a$ and $a^{\dagger}$ (see \cite{AccBouCOSA} and
\cite{ABLT}).

\end{remark}

\section{The Vacuum Characteristic Function of  $\frac{1}{2}(X^2+P^2)$}

References to the spectrum of the Quantum Harmonic Oscillator
Hamiltonian operator
\begin{equation}
\frac{1}{2}(X^2+P^2)=\frac{s^2}{2}-\frac{d^2}{ds^2}\  ,
\end{equation}
are abundant, see, for example, p.145 of \cite{Richtmyer} and also
\cite{BeSen}, and they typically involve reducing it to the
\textit{number operator} $N:=a\,a^\dagger$ using the formulas of
Remark \ref{rem} above. In Theorems \ref{sa3} and \ref{csp2}
below, we prove the essential self-adjointness of the Hamiltonian
operator and the fact that it has purely point spectrum, with the
use of the formula for its resolvent obtained in Theorem 5 of
\cite{BFspec2}.

\begin{theorem}\label{sa3} The operator
$\frac{1}{2}(X^2+P^2)$ is essentially self-adjoint on
$\mathcal{S}$.
\end{theorem}
\begin{proof} In Theorem 4 of \cite{BFspec2} we showed that
$\frac{1}{2}(X^2+P^2)$ admits a self-adjoint extension. To show
that the self-adjoint extension is unique, i.e. that
$\frac{1}{2}(X^2+P^2)$ is essentially self-adjoint, as in the
proof of Theorem \ref{sa2}, we notice that, by Theorem 5 of
\cite{BFspec2} with $c_1=c_2=0$, the \textit{non--homogeneous
Weber equation}
\begin{equation}\label{vn223}
\frac{1}{2}(X^2+P^2)f\pm i f=g  \ ,
\end{equation}
has a (unique) solution $f\in  \mathcal{S}$ a.e., for each $g \in
C_0^{\infty}(\mathbb{R})$ with  $a=\min \supp (g)$ and $b=\max
\supp (g)$, given by

\begin{equation}\label{drrrr}
f(s) = \left\{
\begin{array}{llr}
  -\int_{a\sqrt{2}}^{s\sqrt{2}} g\left(\frac{w}{\sqrt{2}}\right)
 W(s, w)\,dw
 \   ,&   s\in \lbrack a, b \rbrack   \\
& \\
 0\   ,&  s\notin   \lbrack a, b \rbrack
\end{array}
\right. \  ,
\end{equation}

where

\begin{equation}
W(s, w)=M_2(s\sqrt{2}; \mp i)\,M_1(w; \mp i)-M_1(s\sqrt{2}; \mp
i)\, M_2(w; \mp i) \  ,
\end{equation}

\begin{align}
M_1 (x; y)&=e^{-\frac{x^2}{4}}
{}_1F_1\left(-\frac{y}{2}+\frac{1}{4}; \frac{1}{2}; \frac{x^2}{2}
\right)  \  ,\\
M_2(x; y)&= x e^{-\frac{x^2}{4}}
{}_1F_1\left(-\frac{y}{2}+\frac{3}{4}; \frac{3}{2}; \frac{x^2}{2}
\right) \  ,
\end{align}
and
\begin{equation}
{}_{1}F_1\left(y;  c; x\right)=\sum_{n=0}^\infty \frac{ (y)_n
}{(c)_n }\frac{x^n}{n!} \  ,
\end{equation}
is  \textit{Kummer's confluent hypergeometric function}.
  The proof  that $f\in \mathcal{S}$ a.e. proceeds
along the lines of the proof of Theorem \ref{sa}. The proof of the
fact that such an $f$ is unique follows from the fact that the
solution of the \textit{homogeneous Weber equation}
\begin{equation}\label{vn2233}
\frac{1}{2}(X^2+P^2)f\pm i f=0  \ ,
\end{equation}
is (see Theorem 5 of \cite{BFspec2}),
\begin{equation}
 f(s)=c_1\,M_1(s; \mp i)+c_2 \, M_2(s; \mp i) \ ,
\end{equation}
which, as can be seen by the asymptotic formula (13.1.4) of
\cite{web1}, namely,
\begin{equation}
{}_{1}F_1\left(y;  c; x\right)=\frac{\Gamma (c)}{\Gamma (y)}e^x
x^{y-c} \left(1+O\left(|x|^{-1}  \right)
\right)\,\,,\,\,x>0\,\,,\,\,x\to \infty \  ,
\end{equation}
 is in $\mathcal{S}$ if and only if $c_1=c_2=0$.  Thus
$\frac{1}{2}(X^2+P^2)$ is essentially self-adjoint on
$\mathcal{S}$.
\end{proof}

\begin{theorem}\label{csp2}  The operator  $\frac{1}{2}(X^2+P^2)$ has a purely point
 spectrum consisting of
 the numbers $(2 n-1)\frac{1}{2}$, $n=1, 2, ...$. Moreover, for $g$ in the range of
  $z-\frac{1}{2}(X^2+P^2)$, the
resolvent operator $R\left(z; \frac{1}{2}(X^2+P^2)\right)$, is

\begin{align}\label{rops4}
R&\left(z; \frac{1}{2}\left(X^2+P^2\right)\right)
g(s)=c_1(z)\,M_1(s\sqrt{2}; z)+c_2(z) \, M_2(s\sqrt{2}; z)\\
&+\int_{-\infty}^{s\sqrt{2}} g\left(\frac{w}{\sqrt{2}}\right)
\left( M_1(w; z) M_2(s\sqrt{2}; z)-M_1(s\sqrt{2}; z)M_2(w; z)
\right)\,dw\  , \notag
\end{align}
where   $c_1(z),c_2(z)\in \mathbb{C} $,
\begin{align}
M_1 (x; y)&=e^{-\frac{x^2}{4}}
{}_1F_1\left(-\frac{y}{2}+\frac{1}{4}; \frac{1}{2}; \frac{x^2}{2}
\right)\    ,\\
M_2(x; y)&= x e^{-\frac{x^2}{4}}
{}_1F_1\left(-\frac{y}{2}+\frac{3}{4}; \frac{3}{2}; \frac{x^2}{2}
\right)\  ,
\end{align}
and
\begin{equation}
{}_{1}F_1\left(y;  c; x\right)=\sum_{n=0}^\infty \frac{ (y)_n
}{(c)_n }\frac{x^n}{n!} \  ,
\end{equation}
is  \textit{Kummer's confluent hypergeometric function}, where for
$a\in\mathbb{R}$ and $n\geq 1$, $(a)_n=a(a+1)(a+2)\cdots(a+n-1)$,
while for $n=0$, $(a)_0=1$.

\end{theorem}
\begin{proof} The proof of (\ref{rops4}) can be found in Theorem 5 of
\cite{BFspec2}. As in the proof of Theorem \ref{sa3}, the solution
of the \textit{homogeneous Weber equation}
\begin{equation}\label{vn224}
\left(z-\frac{1}{2}(X^2+P^2)\right)f(s)=0  \ ,
\end{equation}
is (see Theorem 5 of \cite{BFspec2}),
\begin{equation}\label{toc}
 f(s)=c_1(z)\,M_1(s\sqrt{2}; z)+c_2(z) \, M_2(s\sqrt{2}; z) \ .
\end{equation}
For $z=2 n-\frac{3}{2}$, $n=1, 2, ...$, the ${}_1F_1$ function in
the definition of $M_1$ becomes a polynomial, so $M_1\in
\mathcal{S}$. Therefore $M_1(s\sqrt{2}; 2 n-\frac{3}{2})$ is, for
each $n$, an eigenfunction of $\frac{1}{2}(X^2+P^2)$ corresponding
to the eigenvalue $z=2 n-\frac{3}{2}$.

Similarly, for $z=2 n-\frac{1}{2}$, $n=1, 2, ...$, the ${}_1F_1$
function in the definition of $M_2$ becomes a polynomial, so
$M_2\in \mathcal{S}$. Therefore $M_2(s\sqrt{2}; 2 n-\frac{1}{2})$
is, for each $n$, an eigenfunction of $\frac{1}{2}(X^2+P^2)$
corresponding to the eigenvalue $z=2 n-\frac{1}{2}$.

Thus the numbers $z=(2 n-1)\frac{1}{2}$, $n=1, 2, ...$, are in the
point spectrum of  $\frac{1}{2}(X^2+P^2)$. They are the only
eigenvalues of $\frac{1}{2}(X^2+P^2)$, since in all other cases
the function $f$ in (\ref{toc}) is not in $\mathcal{S}$. As shown
in Theorem 7.15 of \cite{GS}, $\frac{1}{2}(X^2+P^2)$ has no
continuous spectrum, and there is no residual spectrum since
$\frac{1}{2}(X^2+P^2)$ is essentially self-adjoint. Thus
$\frac{1}{2}(X^2+P^2)$ has a purely point spectrum.

Using the formulas in Chapter 13  of \cite{zj}, we can write the
eigenfunctions in terms of the Hermite polynomials, a well-known
fact in the literature, see, for example,  p.145 of
\cite{Richtmyer} and also \cite{BeSen}.
\end{proof}

\begin{theorem}\label{Ham} For $t\in
\mathbb{R}$, the vacuum characteristic function of
$\frac{1}{2}\left(X^2+P^2\right)$ is
\begin{equation}
\langle  e^{ \frac{i t}{2}\left(X^2+P^2\right)}
\rangle=e^{\frac{it}{2}} \ .
\end{equation}
\end{theorem}
\begin{proof}
 Using formula
(\ref{rops4}) for the resolvent of
$T:=\frac{1}{2}\left(X^2+P^2\right)$, as in the proof of Theorem
\ref{XP+PX}, using the fact that
\begin{equation}
\int_{-\infty}^{s\sqrt{2}} g\left(\frac{w}{\sqrt{2}}\right) \left(
M_1(w; z) M_2(s\sqrt{2}; z)-M_1(s\sqrt{2}; z)M_2(w; z) \right)\,dw
\ ,
\end{equation}
is real, we have
\begin{align}
\langle  e^{itT} \rangle&=\frac{1}{\pi }\int_{-\infty}^\infty e^{i
t x} \, {\rm Im }\,\int_{-\infty}^\infty\Phi(s) R(x; T)\Phi(s)\,ds
\,dx
\\
&= \frac{1}{\pi }\int_{-\infty}^\infty e^{i t x} \, {\rm Im
}\,\int_{-\infty}^\infty\Phi(s) \left( c_1(x)\,M_1(s\sqrt{2};
x)+c_2(x) \, M_2(s\sqrt{2}; x)   \right)\,ds \,dx \  .  \notag
\end{align}
Using, as in the proof of Theorem 6 of \cite{BFspec2}, the fact
that
\begin{equation}
\int_{\mathbb{R}}s^{2n}e^{-s^2}\,
ds=\Gamma\left(n+\frac{1}{2}\right)\,,\,\int_{\mathbb{R}}s^{2n+1}e^{-s^2}\,
ds=0\,,\,n\in\{0, 1, 2,...\} \ ,
\end{equation}
and
\begin{equation}
\Gamma\left(n+\frac{1}{2}\right)=\frac{1 \cdot 3  \cdot \cdot
\cdot (2 n-1) }{2^n }\pi^{1/2} \,\,,\,\,n=1, 2, ... \ ,
\end{equation}
 we see that
\begin{equation}
 \int_{-\infty}^\infty\Phi(s)\, M_2(s\sqrt{2}; x)   \,ds  =0 \ ,
\end{equation}
so,
\begin{equation}\label{gb}
\langle  e^{itT} \rangle=\frac{1}{\pi }\int_{-\infty}^\infty e^{i
t x} \, {\rm Im }\,\int_{-\infty}^\infty\Phi(s) \left(
c_1(x)\,M_1(s\sqrt{2}; x)\right)\,ds \,dx   \ .
\end{equation}
To compute $c_1(x)$, letting $t=0$ in (\ref{gb}) we obtain
\begin{align}\label{gb2}
1&=\frac{1}{\pi }\int_{-\infty}^\infty  \, {\rm Im }\left(
c_1(x)\right)\,\int_{-\infty}^\infty\Phi(s) \,M_1(s\sqrt{2};
x)\,ds \,dx \\
&= \frac{1}{\pi^{5/4} } \int_{-\infty}^\infty  \, {\rm Im }\left(
c_1(x)\right)\,\int_{-\infty}^\infty  e^{-s^2}   \sum_{n=0}^\infty
\frac{ \left(\frac{1}{4}-\frac{x}{2} \right)_n }{
\left(\frac{1}{2}\right)_n\, n!} \,ds \,dx \notag \\
&=\frac{1}{\pi^{5/4} } \int_{-\infty}^\infty  \, {\rm Im }\left(
c_1(x)\right)\,  \sum_{n=0}^\infty \frac{
\left(\frac{1}{4}-\frac{x}{2} \right)_n }{
\left(\frac{1}{2}\right)_n\, n!} \,\Gamma\left(n+\frac{1}{2}\right) \,dx \notag \\
&=\frac{1}{\pi^{3/4} } \int_{-\infty}^\infty  \, {\rm Im }\left(
c_1(x)\right)\,  \sum_{n=0}^\infty \frac{
\left(\frac{1}{4}-\frac{x}{2} \right)_n }{
 n!} \,dx \notag  \  .
\end{align}

As pointed out in the proof of  Theorem 6 of \cite{BFspec2}, the
partial sums of the series

\begin{equation}
 \sum_{n=0}^\infty \frac{
\left(x \right)_n }{n!} \  ,
\end{equation}
are
\begin{equation}
s_k=\sum_{n=0}^k \frac{ \left(x \right)_n }{n!}=\frac{(1+k)
\Gamma(1+x+k)}{ \Gamma (x+1) \Gamma (2+k)} \ .
\end{equation}
Thus,
\begin{equation}
\sum_{n=0}^\infty \frac{ \left(x \right)_n
}{n!}=\lim_{k\to\infty}s_k=\left\{
\begin{array}{llr}
  0\   ,&  x< 0   \\
& \\
1\ , & x=0\\
& \\
 \infty \   ,&  x> 0
 \end{array}
 \right.
\  .
\end{equation}
Thus, in order for (\ref{gb2}) to hold, we must interpret $c_1(x)$
in the distribution sense, as
\begin{equation}
{\rm Im }\left( c_1(x)\right)=\pi^{3/4} \delta_{1/2}\left(x
\right)\  ,
\end{equation}
in which case, (\ref{gb}) implies that
\begin{equation}\label{gb3}
\langle  e^{itT} \rangle=\frac{1}{\pi^{3/4} }\int_{-\infty}^\infty
e^{i t x} \pi^{3/4}  \delta_{1/2}\left(x\right) \sum_{n=0}^\infty
\frac{ \left(\frac{1}{4}-\frac{x}{2} \right)_n }{n!}  \,dx =
e^{\frac{i t}{2}}  \ ,
\end{equation}
which implies that the probability distribution of
 $\frac{1}{2}\left(X^2+P^2\right)$ is degenerate.
\end{proof}

\bibliographystyle{amsplain}

\end{document}